\pdfoutput=1
\newif\ifFull
\Fullfalse
\documentclass[oribibl,10pt]{llncs}
\usepackage{hyperref}
\usepackage{multibib}
\newcites{appendix}{Additional References}

\usepackage{graphicx}\graphicspath{{figures/}}
\usepackage{enumerate}
\usepackage{times}
\usepackage{cite}
\usepackage{amsmath, amssymb, color}
\usepackage{todonotes}

\let\doendproof\endproof
\renewcommand\endproof{~\hfill\qed\doendproof}

\newcommand{\reals}{\mathbf{R}}
\newcommand{\rationals}{\mathbf{Q}}
\newcommand{\integers}{\mathbf{Z}}
\newcommand{\Emph}[1]{\smallskip\textbf{\textit{#1}}~}
\DeclareMathOperator{\Gal}{Gal}
\DeclareMathOperator{\Bipyr}{Bipyramid}
\DeclareMathOperator{\Aut}{Aut}
\DeclareMathOperator{\Pack}{Pack}

\DeclareMathOperator{\splitting}{split}
\DeclareMathOperator{\adjacency}{adj}
\DeclareMathOperator{\degree}{deg}
\DeclareMathOperator{\laplacian}{lap}
\DeclareMathOperator{\rlaplacian}{lap_\rho}
\DeclareMathOperator{\transition}{tran}
\DeclareMathOperator{\characteristic}{char}

\DeclareMathOperator{\dist}{dist}

\newcommand{\nth}[1]{$#1^\text{th}$}

\title{The Galois Complexity of Graph Drawing:\\ 
Why Numerical Solutions are Ubiquitous
for Force-Directed, Spectral, and Circle Packing Drawings%
\thanks{This research was supported in part by ONR MURI grant N00014-08-1-1015 and NSF grants 1217322, 1011840, and 1228639.}}
\author{Michael J. Bannister \and William E. Devanny \and\\ David Eppstein \and Michael T. Goodrich}
\institute{Department of Computer Science, University of California, Irvine}
  
\begin{document}
\maketitle

\begin{abstract}
Many well-known graph drawing techniques, 
including force directed drawings, spectral graph layouts, multidimensional scaling, and circle packings, 
have algebraic formulations. However, practical methods for producing such
drawings ubiquitously use iterative numerical approximations rather than constructing and then solving algebraic expressions representing their exact solutions.
To explain this phenomenon, we use Galois theory to show that many variants of these problems have solutions that cannot be expressed by nested radicals or nested roots of low-degree polynomials.
Hence, such solutions cannot be computed exactly even in extended computational
models that include such operations.
\end{abstract}

\setcounter{page}{1}
\pagestyle{plain}

\section{Introduction}
One of the most powerful paradigms for drawing a graph 
is to construct an algebraic formulation for a suitably-defined optimal drawing
of the graph and then solve this formulation to produce a drawing.
Examples of this \emph{algebraic graph drawing} approach include the 
force-directed, spectral, multidimensional scaling, and circle packing
drawing techniques 
(which we review in Appendix~\ref{app:gd} for readers unfamiliar
with them).

Even though this paradigm starts from an algebraic formulation,
the ubiquitous method
for solving such formulations is to approximately optimize them numerically in
an iterative fashion.
That is, with a few exceptions for linear systems~\cite{Chrobak:1996,hk-gaa-92,tutte1963draw},
approximate numerical solutions for algebraic graph drawing are overwhelmingly
preferred over exact symbolic solutions.
It is therefore natural to ask if this preference for numerical solutions 
over symbolic solutions is
inherent in algebraic graph drawing 
or due to some other phenomena, such as 
laziness or lack of mathematical sophistication on the part of those who are
producing the algebraic formulations.

In this paper, we introduce a framework for deciding whether certain 
algebraic graph drawing formulations have symbolic solutions, and we show that
exact symbolic solutions are, in fact, impossible in several 
algebraic computation models, for some simple examples
of common algebraic graph drawing formulations, including 
force-directed graph
drawings (in both the Fruchterman--Reingold~\cite{FruRei-SPE-1991} and 
Kamada--Kawai~\cite{KamKaw-IPL-1989} approaches), 
spectral graph drawings~\cite{Kor-CMA-2005}, classical multidimensional scaling~\cite{KruSee-CSG-80}, and 
circle packings~\cite{Koe-BSAWL-36}.
Note that these impossibility 
results go beyond saying that such symbolic solutions are
computationally infeasible or undecidable to find---instead, 
we show that such solutions do not exist.

To prove our results, we use Galois theory, a connection between the theories of algebraic numbers and abstract groups. 
Two classical applications of Galois theory use it to prove the impossibility of the classical Greek problem of doubling the cube using compass and straightedge, and of solving fifth-degree polynomials by nested radicals. In our terms, these results concern quadratic computation trees and radical computation trees, respectively. 
Our proofs build on this theory by applying Galois theory to the algebraic numbers given by the vertex positions in different types of graph drawings. For force-directed and spectral drawing, we find small graphs (in one case as small as a length-three path) whose drawings directly generate unsolvable Galois groups. For circle packing, an additional argument involving the compass and straightedge constructibility of M\"obius transformations allows us to transform arbitrary circle packings into a canonical form with two concentric circles, whose construction is equivalent to the calculation of certain algebraic numbers.
Because of this mathematical foundation, we refer to this topic as the 
\emph{Galois complexity} of graph drawing.

\Emph{Related Work.}
The problems for which
Galois theory has been used to prove unsolvability in
simple algebraic computational models include
shortest paths around polyhedral obstacles~\cite{Baj-Pol-1985}, shortest paths through weighted regions of the plane~\cite{CarGriMah-2013},
the geometric median of planar points~\cite{Baj-DCG-1988},
computing structure from motion in computer vision~\cite{NisHarSte-CVPR-2007},
and finding polygons of maximal area with specified edge lengths~\cite{Var-SB-2004}.
In each of these cases, the non-existence of a nested radical formula for the solution is established by finding a Galois group containing a symmetric group of constant degree at least five. In our terminology, this shows that these problems cannot be solved by a radical computation tree.
We are not aware of any previous non-constant lower bounds on the degree of the polynomial roots needed to solve a problem, comparable to our new bounds using the root computation tree model.
Brightwell and Scheinerman~\cite{BriSch-SJDM-93} show that some 
circle packing graph representations cannot be constructed by compass and straightedge
(what we call the quadratic computation tree model).

\section{Preliminaries}

\Emph{Models of computation.}
\ifFull
\else
We
\fi
define models of computation based on the \emph{algebraic computation tree}~\cite{Ben-STOC-83,Yao-SJC-91}, in which each node computes a value or makes a decision using standard arithmetic functions of previously computed values.
{} Specifically, we define the following variant models:
\begin{itemize}
\item 
A \emph{quadratic computation tree} is an algebraic computation tree in which the set of allowable functions for each computation node is augmented with square roots and complex conjugation. These trees capture the geometric constructions that can be performed by compass and unmarked straightedge.
\item 
A \emph{radical computation tree} is an algebraic computation tree in which the 
set of allowable functions is augmented with the \nth{k} root operation, 
where $k$ is an integer parameter to the operation, and with complex conjugation. These trees capture the calculations whose results can be expressed as nested radicals.
\item 
A \emph{root computation tree} is an algebraic computation tree in which the allowable functions include the ability to find complex roots of polynomials whose coefficients are integers or previously computed values, and to compute complex conjugates of previously computed values. 
For instance, this model can compute any algebraic number.
As a measure of complexity in this model, we define
the \emph{degree} of a root computation tree as the maximum degree of any of its polynomials.
A \emph{bounded-degree root computation tree}
has its degree bounded by some constant unrelated to the size of its input.
Thus, a quadratic computation tree is exactly a 
bounded-degree root computation tree (of degree two).
\end{itemize}
Our impossibility results and degree lower bounds for these models 
imply the same 
results for algorithms in more realistic models of computation that use as a black box the corresponding primitives 
for constructing and representing algebraic numbers in
symbolic computation systems.
Because our results are lower bounds, they also apply \textit{a fortiori} to weaker primitives, such as systems limited to \emph{real} algebraic numbers, which
don't include complex conjugation.

It is important to note that each of
the above models can generate algebraic numbers of unbounded degree. 
For instance, even the quadratic computation tree 
(compass and straightedge model) can construct regular $2^k$-gons, whose coordinates are
algebraic numbers with degrees that are high powers of two.
Thus, to prove lower bounds and impossibility results in these models, 
it is not sufficient to prove that a problem is described by a high-degree polynomial; additional structure is needed.

\Emph{Algebraic graph theory.}
In algebraic graph theory, the properties of a graph are examined via the spectra of several matrices associated with the graph. The \emph{adjacency matrix} $A = \adjacency(G)$ of a graph $G$ is the $n\times n$ matrix with $A_{i,j}$ equal to $1$ if there is an edge between $i$ and $j$ and $0$ otherwise. The \emph{degree matrix} $D = \degree(G)$ of $G$ is the $n\times n$ matrix with $D_{i,i} = \deg(v_i)$. From these two matrices we define the \emph{Laplacian matrix}, $L = \laplacian(G) = D - A$, 
and the \emph{transition matrix}, $T = \transition(G) = D^{-1}A$.

\begin{lemma}\label{lem:regular-eigenvalues}
For a regular graph $G$,  $ \adjacency(G)$, $ \laplacian(G)$, and $ \transition(G)$ have the same set of eigenvectors.
\end{lemma}

\begin{lemma}\label{lem:cycle-eigenvalues}
For the cycle on $n$ vertices, the eigenvalues of $\adjacency(G)$ are $2\cos(2\pi k/n)$, for $0 \leq k < n$.
\end{lemma}

\Emph{M\"{o}bius transformations.}
We may represent each point $p$ in the plane by a complex number,
$z$,
whose real part represents $p$'s $x$ coordinate and whose 
imaginary part represents $p$'s $y$ coordinate.
A \emph{M\"obius transformation} is a fractional linear transformation,
\ifFull
\else
$z\mapsto (az+b)/(cz+d)$,
\fi
defined by a 4-tuple $(a,b,c,d)$ of complex numbers, or the complex conjugate of such a transformation.
We prove the following in Appendix~\ref{app:proofs}.

\begin{lemma}\label{lem:mob-trans}
Given any two disjoint circles, a M\"{o}bius transformation mapping them to two concentric circles can be constructed using a quadratic computation tree.
\end{lemma}

\Emph{Number theory.}
The \emph{Euler totient function}, $\phi(n)$, counts the number of integers in the interval $[1,n-1]$ that are relatively prime to $n$. It can be calculated from the prime factorization $n = \prod p_i^{r_i}$ by the formula \[
\phi(n) = \prod p_i^{r_i-1}(p_i-1).
\]

A \emph{Sophie Germain prime} is a prime number $p$ such that $2p+1$ 
is also prime~\cite{Sho-CINTA-09}. 
It has been conjectured that there are infinitely many of them, but the conjecture remains unsolved.
The significance of these primes for us is that, when $p$ is a Sophie Germain prime, $\phi(2p+1)$ has the large prime factor~$p$. 
An easy construction gives a number $n$
for which $\phi(n)$ has a prime factor of size $\Omega(\sqrt n)$: simply let $n=p^2$ for a prime $p$, with $\phi(n)=p(p-1)$. 
Baker and Harman~\cite{BakHar-AA-98} proved the following stronger bound.

\begin{lemma}[Baker and Harman~\cite{BakHar-AA-98}]
\label{lem:large-prime-factor-of-phi}
For infinitely many prime numbers $p$, the largest prime factor of $\phi(p)$  is at least $p^{0.677}$.
\end{lemma}

\Emph{Field theory.}
A \emph{field} is a system of values and arithmetic operations over them (addition, subtraction, multiplication, and division) obeying similar axioms to those of rational arithmetic, real number arithmetic, and complex number arithmetic: addition and multiplication are commutative and associative, multiplication distributes over addition, subtraction is inverse to addition, and division is inverse to multiplication by any value except zero. A field $K$ is an \emph{extension} of a field $F$, and $F$ is a \emph{subfield} of $K$ (the \emph{base field}), if the elements of $F$ are a subset of those of $K$ and the two fields' operations coincide for those values. $K$ can be viewed as a vector space over $F$ (values in $K$ can be added to each other and multiplied by values in~$F$) and the \emph{degree} $[K:F]$ of the extension is its dimension as a vector space.
For an element $\alpha$ of $K$ the notation $F(\alpha)$ represents the set of values that can be obtained from rational functions (ratios of univariate polynomials) with coefficients in $F$ by plugging in $\alpha$ as the value of the variable. $F(\alpha)$~is itself a field, intermediate between $F$ and $K$. In particular, we will frequently consider field extensions $\rationals(\alpha)$ where $\rationals$ is the field of rational numbers and $\alpha$ is an \emph{algebraic number}, the complex root of a polynomial with rational coefficients.

\begin{lemma}\label{lem:extensions}
If $\alpha$ can be computed by a root computation tree of degree $f(n)$, then $[\rationals(\alpha):\rationals]$ is $f(n)$-smooth, i.e., it has no prime factor  $>f(n)$. In particular, if $\alpha$ can be computed by a quadratic computation tree, then $[\rationals(\alpha):\rationals]$ is a power of two.
\end{lemma}

\begin{proof}
See Appendix~\ref{app:proofs}.
\end{proof}

A \emph{primitive root of unity} $\zeta_n$ is a root of
$x^n-1$ whose powers give all  other roots of the same polynomial. As a complex number we can take
$\zeta_n=\exp(2i\pi/n)$.

\begin{lemma}[Corollary~9.1.10 of \cite{Cox2012}, p.~235]
\label{lem:roots-of-unity}
$[\rationals(\zeta_n):\rationals] = \phi(n)$.
\end{lemma}

\Emph{Galois theory.}
A \emph{group} is a system of values and a single operation (written as multiplication) that is associative and in which every element has an inverse. The set of permutations of the set $[n] = \{1,2,\ldots,n\}$, multiplied by function composition, is a standard example of a group and is denoted by $S_n$. A \emph{permutation group} is a subgroup of $S_n$; i.e., it is a set of permutations that is closed under the group operation.

A \emph{field automorphism} of the field $F$ is a bijection $\sigma :F \to F$ that respects the field operations, i.e., $\sigma(xy) = \sigma(x)\sigma(y)$ and $\sigma(x+y) = \sigma(x) + \sigma(y)$. The set of all field automorphism of a field $F$ forms a group denoted by $\Aut(F)$. Given a field extension $K$ of $F$, the subset of $\Aut(K)$ that leaves $F$ unchanged is itself a group, called the \emph{Galois group} of the extension, and is denoted
\[
\Gal(K/F) = \{\sigma \in \Aut(K) \mid \sigma(x) = x \text{ for all } x \in F \}.
\]
The \emph{splitting field} of a polynomial, $p$, with rational coefficients, denoted $\splitting (p)$ is the smallest subfield of the complex numbers that contains  all the roots of the polynomial. 
Each automorphism in $\Gal(\splitting(p)/\rationals)$ permutes the roots of the polynomial, no two automorphisms permute the roots in the same way, and these permutations form a group, so $\Gal(\splitting(p)/\rationals)$ can be thought of as a permutation group.

\begin{lemma}\label{lem:radical-tree}
If $\alpha$ can be computed by a radical computation tree and $K$ is the splitting field of an irreducible polynomial with $\alpha$ as one of its roots, then $\Gal(K/\rationals)$ does not contain $S_n$ as a subgroup for any $n \geq 5$.
\end{lemma}
\begin{proof}
If $\alpha$ is computable by a radical computation tree, it can be written as an expression using nested radicals. If $K$ is the splitting field of an irreducible polynomial with such an expression as a root, $\Gal(K/\rationals)$ is a solvable group (Def.~8.1.1 of \cite{Cox2012}, p. 191 and Theorem~8.3.3 of \cite{Cox2012}, p. 204). But $S_n$ is not solvable for $n \geq 5$ (Theorem 8.4.5 of \cite{Cox2012}, p. 213), and every subgroup of a solvable group is solvable (Proposition 8.1.3 of \cite{Cox2012}, p. 192). Thus, $\Gal(K/\rationals)$ cannot contain $S_n$ ($n\ge 5$) as a subgroup.\end{proof}

The next lemma allows us to infer properties of a Galois group from the coefficients of a \emph{monic} polynomial, that is, a polynomial with integer coefficients whose first coefficient is one. The \emph{discriminant} of a monic polynomial is (up to sign) the product of the squared differences of all pairs of its roots; it can also be computed as a polynomial function of the coefficients. The lemma is due to Dedekind and proven in \cite{Cox2012}.

\begin{lemma}[Dedekind's theorem]
Let $f(x)$ be an irreducible monic polynomial in $\integers[x]$ and $p$ a prime not dividing the discriminant of $f$. If $f(x)$ factors into a product of irreducibles of degrees $d_0, d_1, \ldots d_r$ over $\integers/p\integers$, then $\Gal(\splitting(f)/\rationals)$ contains a permutation that is the composition of disjoint cycles of lengths $d_0, d_1, \ldots, d_r$.
\end{lemma}

A permutation group is \emph{transitive} if, for every two elements $x$ and $y$ of the elements being permuted, the group includes a permutation that maps $x$ to $y$. If $K$ is the splitting field of an irreducible polynomial of degree $n$, then $\Gal(K/\rationals)$ (viewed as a permutation group on the roots) is necessarily transitive. The next lemma allows us to use Dedekind's theorem to prove that $\Gal(K/\rationals)$ equals $S_n$. It is a standard exercise in abstract algebra (e.g., \cite{Jac2012}, Exercise~3, p.~305).

\begin{lemma}
\label{lem:swap+cycle}
If a transitive subgroup  $G$ of $S_n$ contains a transposition and an $(n-1)$-cycle, then $G = S_n$.
\end{lemma}

\section{Impossibility Results for Force Directed Graph Drawing}
In the Fruchterman and Reingold~\cite{FruRei-SPE-1991} 
force-directed model, each vertex is pulled toward its neighbors with an attractive force, $f_a(d) = d^2/k$, and pushed away from all vertices with a repulsive force, $f_r(d) = k^2/d$. The parameter $k$ is a constant that sets the scale of the drawing, and $d$ is the distance between vertices. We say that a drawing is a Fruchterman and Reingold equilibrium when the total force at each vertex is zero.

In the Kamada and Kawai~\cite{KamKaw-IPL-1989} 
force-directed model, every two vertices are connected by a spring with rest length and spring constant determined by the structure of the graph. The total energy of the graph is defined to be
\[
    E = \sum_i \sum_{j > i} \frac{1}{2} k_{ij} \big(\dist(p_i,p_j) - \ell_{ij} \big)^2,
\]
where
\ifFull
\else
$p_i = \text{position of vertex $v_i$}$,
$d_{ij} = \text{graph theoretic distance between $v_i$ and $v_j$}$,
$L = \text{a scaling constant}$,
$\ell_{ij} = Ld_{ij}$,
$K = \text{a scaling constant}$,
and
$k_{ij} = K/d_{i,j}^2$.
\fi
We say that a drawing is a Kamada--Kawai equilibrium if $E$ is at a local minimum. The necessary conditions for such a local minimum are as follows:
\begin{align*}
\frac{\partial E}{\partial x_j} &= \sum_{i \neq j} k_{ji}
(x_j - x_i) \left(1 - \frac{\ell_{ji}}{\dist(p_j, p_i)}\right)=0 && 1 \leq j \leq n\\
\frac{\partial E}{\partial y_j} &= \sum_{i \neq j} k_{ji}
(y_j - y_i) \left(1 - \frac{\ell_{ji}}{\dist(p_j, p_i)}\right)=0 && 1 \leq j \leq n.
\end{align*}

For either of these approaches to
force-directed graph drawing, a graph can have multiple equilibria (Figure~\ref{fig:2fr-configs}).
In such cases, typically, one equilibrium is the ``expected'' drawing of the graph and others represent undesired drawings that are not likely to be found by the drawing algorithm. To make the positions of the vertices in this drawing concrete, we assume that the constants~$k$ (Fruchterman--Reingold), $L$, and $K$ (Kamada--Kawai) are all equal to 1.
As we will demonstrate, there exist graphs whose expected drawings cannot be constructed in our models of computation.
Interestingly, the graphs we use for these results
are not complicated configurations unlikely to 
arise in practice, but are instead graphs so simple 
that they might at first
be dismissed as insufficiently challenging even 
to be used for debugging purposes.

\begin{figure}[tb]
\begin{minipage}[t]{0.4\textwidth} 
\centering
\includegraphics[scale=0.4]{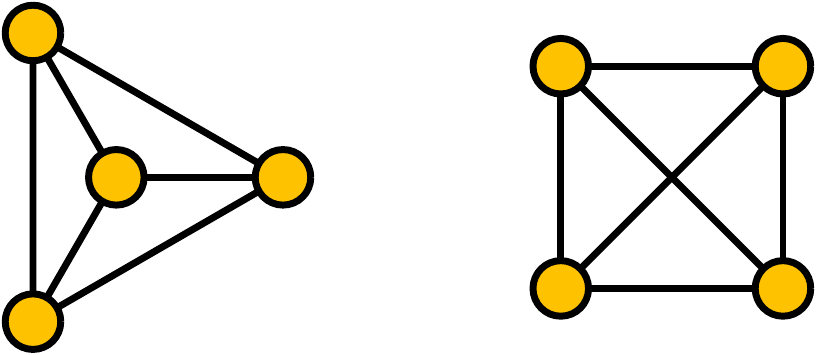}
\caption{Two stable drawings of $K_4$.}
\label{fig:2fr-configs}
\end{minipage}\hfill
\begin{minipage}[t]{0.4\textwidth} 
\centering
\includegraphics[scale=0.4]{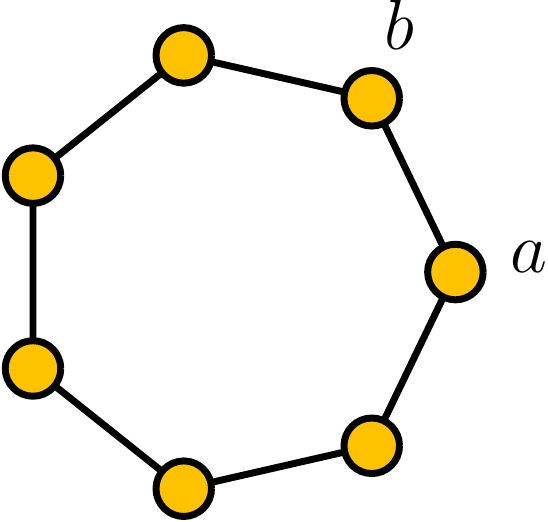}
\caption{A drawing whose coordinates cannot be computed by a quadratic computation tree.}
\label{fig:seven-cycle}
\end{minipage}
\end{figure}

\Emph{Root computation trees.}
Consider the cycle $C_n$ with $n$ vertices. When drawn with force directed algorithms, either Fruchterman and Reingold or Kamada and Kawai, the embedding typically places all vertices equally spaced on a circle, such that neighbors are placed next to each other, as shown in Figure~\ref{fig:seven-cycle}. As an easy warm-up to our main results, we observe that this is not always possible using a quadratic computation tree.

\begin{theorem}\label{thm:force-quad}
There exist a graph with seven vertices such that it is not possible in a quadratic computation tree to compute the coordinates of every possible Fruchterman and Reingold equilibrium or
every possible Kamada and Kawai equilibrium.
\end{theorem}
\begin{proof}
Let $G$ be the cycle $C_7$ on seven vertices. Both algorithms have the embedding shown in Figure~\ref{fig:seven-cycle} (suitably scaled) as an equilibrium. In this embedding let $a$ and $b$ be two neighboring vertices and $\alpha$ and $\beta$ their corresponding complex coordinates. Then $\alpha / \beta$ is equal to $\pm \zeta_7$ the seventh root of unity. By Lemma~\ref{lem:roots-of-unity}
\[
[\rationals(\zeta_7) : \rationals] = \phi(7) = 6.
\]
Since $6$ is not a power of two, Lemma~\ref{lem:extensions} implies that $\zeta_7$ cannot be constructed by a quadratic computation tree. Therefore, neither can this embedding.
\end{proof}

\begin{theorem}\label{thm:force-root}
For arbitrarily large values of $n$, there are graphs on $n$ vertices such that constructing the coordinates of all Fruchterman and Reingold equilibria on a root computation tree requires degree $\Omega(n^{0.677})$. If there exists infinitely many Sophie Germain primes, then there are graphs for which computing the coordinates of any Fruchterman and Reingold equilibria requires degree $\Omega(n)$. The same results with the same graphs hold for Kamada and Kawai equilibria.
\end{theorem}
\begin{proof}
As in the previous theorem we consider embedding cycles with their canonical embedding, which is an equilibrium for both algorithms. The same argument used in the previous theorem shows we can construct $\zeta_n$ from the coordinates of the canonical embedding of the cycle on $n$ vertices.

We consider cycles with $p$ vertices where $p$ is a prime number for which $\phi(p) = p - 1$ has a large prime factor $q$. If arbitrarily large Sophie Germain primes exist we let $q$ be such a prime and let $p=2q+1$. Otherwise, by Lemma~\ref{lem:large-prime-factor-of-phi} we choose $p$ in such a way that its largest prime factor $q$ is at least $p^{0.677}$. Now, by Lemma~\ref{lem:roots-of-unity} we have:
\[
[\rationals(\zeta_p) : \rationals] = \phi(p) = p - 1.
\]
This extension is not $D$-smooth for any $D$ smaller than $q$, and therefore every construction of it on a root computation tree requires degree at least $q$.
\end{proof}

Thus, such drawings are not possible on a bounded-degree root computation tree.

\Emph{Radical computation trees.}
To show that the coordinates of a Fruchterman and Reingold equilibrium are in general not computable with a radical computation tree we consider embedding the path with three edges, shown in Figure~\ref{fig:fr-path}. 
We assume that all of the vertices are embedded colinearly and without edge or vertex overlaps. These assumptions correspond to the equilibrium that is typically produced by the Fruchterman and Reingold algorithm.

\begin{figure}[tb]
\begin{minipage}[t]{0.4\textwidth} 
\centering
\includegraphics[scale=0.4]{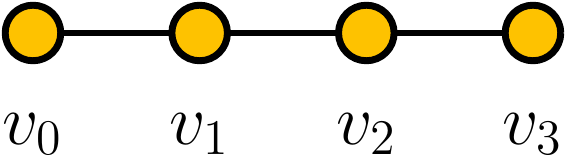}
\caption{A graph whose Fruchterman--Reingold coordinates cannot be computed by a radical computation tree.}
\label{fig:fr-path}
\end{minipage}\hfill
\begin{minipage}[t]{0.4\textwidth} 
\centering
\includegraphics[scale=0.4]{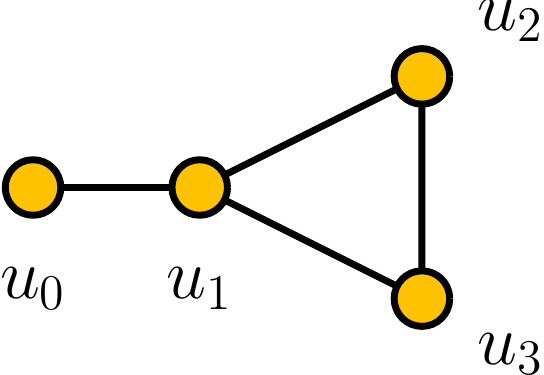}
\caption{A graph whose Kamada--Kawai coordinates cannot be computed by a radical computation tree.}
\label{fig:kk-radical}
\end{minipage}
\end{figure}

Let $a > 0$ be the distance from $v_0$ to $v_1$ (equal by symmetry to the distance from $v_2$ to $v_3$) and let $b > 0$ be the distance from $v_1$ to $v_2$. We can then express the sum of all the forces at vertex $v_0$ by the equation
\[
F_0 = a^2
-\frac{1}{a}
-\frac{1}{a+b}
-\frac{1}{2a+b}
=
\frac{2a^5 + 3a^4b + a^3b^2 - 5a^2 - 5ab - b^2}{2a^3 + 3a^2b + ab^2},
\]
and the sum of all the forces at vertex $v_1$ by the equation
\[
F_1 = 
-a^2
+\frac{1}{a}
+b^2
-\frac{1}{b}
-\frac{1}{a+b}
=
\frac{-a^4b - a^3b^2 + a^2b^3 - a^2 + ab^4 - ab + b^2}{a^2b + ab^2}.
\]
In an equilibrium state we have $F_1 = F_2 = 0$. Equivalently, the numerator $p$ of $F_1$ and the numerator $q$ of $F_2$ are both zero, where
\begin{align*}
p(a,b) &= 2a^5 + 3a^4b + a^3b^2 - 5a^2 - 5ab - b^2 = 0 \\
q(a,b) &= -a^4b - a^3b^2 + a^2b^3 - a^2 + ab^4 - ab + b^2 = 0.
\end{align*}
To solve this system of two equations and two unknowns we can eliminate variable~$a$ and produce the following polynomial, shown as a product of irreducible polynomials, whose roots give the values of $b$ that lead to a solution.
\[
\frac{1}{3} b^2 (3b^{15} - 48b^{12} + 336b^9 - 1196b^6 + 1440b^3 + 144).
\]
The factor $b^2$ corresponds to degenerate drawings and may safely be eliminated.
Let $f$ be the degree-fifteen factor; then $f(x)=g(x^3)$ for a quintic polynomial~$g$. A radical computation tree can compute the roots of $f$ from the roots of $g$, so we need only show that the roots of $g$ cannot be computed in a radical computation tree. To do this, we convert $g$ to a monic polynomial $h$ with the same splitting field, via the transformation
\[
h(x) =\frac{x^5}{144} g(6/x) = x^5 + 60x^4 - 299x^3 + 504x^2 - 432x + 162.
\]
The polynomial $h$ can be shown to be irreducible by manually verifying that it has no linear or quadratic factors. Its discriminant is $-2^6 \cdot 3^9 \cdot 2341^2 \cdot 2749$, and $h$ factors modulo primes $5$ and $7$ (which do not divide the discriminant) into irreducibles:
\begin{align*}
h(x) \equiv (x + 1)  (x^4 + 3x^3 + 6x^2 + x + 1)&\pmod{7}\\
h(x) \equiv (x^2 + 3x + 4)  (x^3 + 2x^2 + x + 3) &\pmod{5}.
\end{align*}
By Dedekind's theorem, the factorization modulo 7 implies the existence of a $4$-cycle in $\Gal(\splitting(h)/\rationals)$, and the factorization modulo $5$ implies the existence of a permutation that is the composition of a transposition and a $3$-cycle.
 Raising the second permutation to the power $3$ yields a transposition. 
By Lemma~\ref{lem:swap+cycle}, $\Gal(\splitting(h)/\rationals) = S_5$. 
So by Lemma~\ref{lem:radical-tree} the value of $b$ cannot be computed by a radical computation tree. 
Thus, we cannot compute the equilibrium coordinates of the path with three edges under the assumptions that the vertices are collinear and there are no vertex or edge overlaps.

\begin{theorem}
There exists a graph on four vertices such that it is not possible on a radical computation tree to construct the coordinates of every possible Fruchterman and Reingold equilibrium.
\end{theorem}

To show that the coordinates of a Kamada and Kawai equilibrium are in general not computable with a radical computation tree we consider the graph depicted in Figure~\ref{fig:kk-radical}.

\begin{theorem}
\label{thm:kk-radical}
There exists a graph on four vertices such that it is not possible on a radical computation tree to construct the coordinates of every possible Kamada and Kawai equilibrium.
\end{theorem}
\begin{proof}
See Appendix~\ref{app:proofs}.
\end{proof}

\section{Impossibility Results for Spectral Graph Drawing}
\Emph{Root computation trees.}
We begin with the following result for root computation trees.

\begin{theorem}
For arbitrarily large values of $n$, there are graphs on $n$ vertices such that constructing spectral graph drawings based on the adjacency, Laplacian, relaxed Laplacian, or transition matrix requires a root computation tree of degree $\Omega(n^{0.677})$. If there exist infinitely many Sophie Germain primes, then there are graphs for which computing these drawings requires degree $\Omega(n)$.
\end{theorem}

\begin{proof}
Since all of the referenced matrices have rational entries, it suffices to consider the computability of their eigenvalues. 
Further, if we restrict our attention to regular graphs it suffices to consider the eigenvalues of just the adjacency matrix, 
$M = \adjacency(G)$, by Lemma~\ref{lem:regular-eigenvalues}. Let $p$ be a prime and $G$ the cycle on $p$ vertices. By Lemma~\ref{lem:cycle-eigenvalues} the eigenvalues of $A = \adjacency(G)$ are given by $2\cos(2\pi k/p)$ for $0 \leq k \leq p-1$.
In a root computation tree of degree at least $2$ the primitive root of unity $\zeta_p = \exp(2i\pi/p)$ can be computed from $2\cos(2\pi k/p)$ for all $k \neq 0$. 
Therefore, from the proof of Theorem~\ref{thm:force-root},
for arbitrarily large~$n$, 
there are graphs on $n$ vertices such that $M$ has one rational eigenvector (for $k=0$) and 
the computation of any other eigenvector on a root computation tree requires degree $\Omega(n^{0.667})$.
If infinitely many Sophie Germain primes exist, there are graphs for which computing these eigenvectors requires degree $\Omega(n)$.
\end{proof}

Thus, such drawings are not possible on a bounded-degree root computation tree.

\Emph{Radical computation trees.} 
To show that in general the eigenvectors associated with a graph are not constructible with a radical tree we consider the graph, $Y$, on nine vertices in Figure~\ref{fig:y-graph} for the Laplacian and relaxed Laplacian matrices, and the graph, $H$, on twelve vertices in Figure~\ref{fig:h-graph} for the adjacency and transition matrices.

\begin{figure}[htb]
\vspace*{-12pt}
\centering
\includegraphics[scale=0.4]{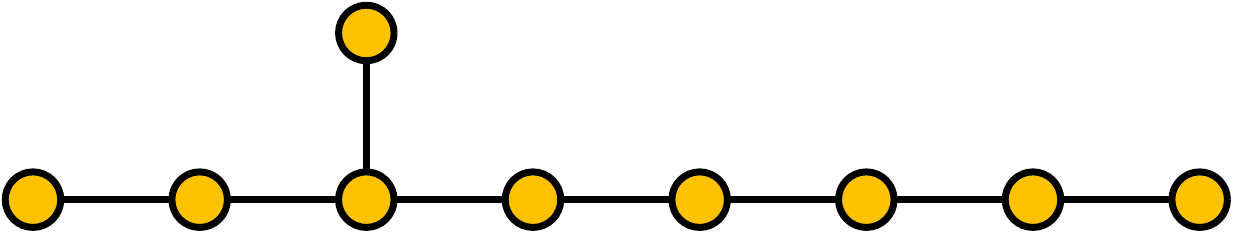}
\caption{A graph $Y$ whose Laplacian eigenvectors are uncomputable by a radical tree.}
\label{fig:y-graph}
\end{figure}

\ifFull
\else
The characteristic polynomial, $ p(x) = \det(M - xI)$,
for the Laplacian matrix for $Y$,
can be computed to be
\fi
\begin{multline*}
p(x) = \characteristic(\laplacian(Y))\\ = x   (x^8 - 16 x^7 + 104 x^6 - 354 x^5 + 678 x^4 - 730 x^3 + 417 x^2 - 110 x + 9).
\end{multline*}

\begin{lemma}[St\"{a}ckel \cite{Sta1918}]\label{lem:2n+1}
If $f(x)$ is a polynomial of degree $n$ with integer coefficients and $\vert f(k) \vert$ is prime for $2n+1$ values of $k$, then $f(x)$ is irreducible.
\end{lemma}

Let $q = p(x) / x$. The polynomial $q$ is irreducible by Lemma~\ref{lem:2n+1}, as it produces a prime number for 17 integer inputs from $0$ to $90$. The discriminant of $q$ is $2^8 \cdot 9931583$ and we have the following factorizations of $q$ modulo the primes $31$ and $41$.
\begin{align*}
p_1(x) \equiv (x + 27)  (x^7 + 19x^6 + 25x^5 + 25x^4 + 3x^3 + 26x^2 + 25x + 21) &\pmod{31} \\
p_1(x) \equiv (x + 1)  (x^2 + 15x + 39)  (x^5 + 9x^4 + 29x^3 + 10x^2 + 36x + 16) &\pmod{41}.
\end{align*}
By Dedekind's theorem, the factorization modulo $31$ implies the existence of a $7$-cycle, and the factorization modulo $41$ implies the existence of a permutation that is the composition of a transposition and a $5$-cycle. The second permutation raised to the fifth power produces a transposition. Thus, Lemma~\ref{lem:swap+cycle} implies $\Gal(\splitting(p_1)/\rationals) = S_8$. So by Lemma~\ref{lem:radical-tree} the only eigenvalue of $\laplacian(Y)$ computable in a radical computation tree is 0.
For the relaxed Laplacian we consider the two variable polynomial $f(x,\rho) = \characteristic(\rlaplacian(Y))$. Since setting $\rho$ equal to $1$ produces a polynomial with Galois group $S_8$, Hilbert's irreducibility theorem tells us that the set of $\rho$ for which the Galois group of $f(x,\rho)$ is $S_8$ is dense in $\rationals$.

\begin{theorem}
There exists a graph on nine vertices such that it is not possible to construct a spectral graph drawing based on the Laplacian matrix in a radical computation tree. For this graph there exists a dense subset $A$ of $\rationals$ such that it is not possible to construct a spectral graph drawing based on the relaxed Laplacian with $\rho \in A$ in a radical computation tree.
\end{theorem}

In Appendix~\ref{app:spectral} we similarly prove that spectral drawings based on the adjacency matrix and the transition matrix cannot be constructed by a radical computation tree.  In Appendix~\ref{app:mds} we similarly prove that drawings produced by classical multidimensional scaling cannot be constructed by a radical computation tree.
\section{Impossibility Results for Circle Packings}

\Emph{Root computation trees.} 
A given graph may be represented by infinitely many circle packings, related to each other by M\"obius transformations. But as we now show, if one particular packing cannot be constructed in our model, then there is no other packing for the same graph that the model can construct.

\begin{lemma}
\label{lem:concentric}
Suppose that a circle packing $P$ contains two concentric circles. Suppose also that at least one radius of a circle or distance between two circle centers, at least one center of a circle, and the slope of at least one line connecting two centers of circles in $P$ can all be constructed by one of our computation models, but that $P$ itself cannot be constructed. Then the same model cannot construct any circle packing that represents the same underlying graph as $P$.
\end{lemma}

\begin{proof}
Suppose for a contradiction that the model could construct a circle packing $Q$ representing the same graph as $P$. By Lemma~\ref{lem:mob-trans} it could transform $Q$ to make the two circles concentric, giving a packing that is similar either to $P$ or to the inversion of $P$ through the center of the concentric circles. By one more transformation it can be made similar to $P$. The model could then rotate the packing so the slope of the line connecting two centers matches the corresponding slope in $P$, scale it so the radius of one of its circles matches the corresponding radius in $P$, and translate  the center of one of its circles to the corresponding center in $P$, resulting in $P$ itself. This gives a construction of $P$, contradicting the assumption.
\end{proof}

We define $\Bipyr (k)$ to be the graph formed by the vertices and edges of a $(k+2)$-vertex bipyramid (a polyhedron formed from two pyramids over a $k$-gon by gluing them together on their bases). In graph-theoretic terms, it consists of a $k$-cycle and two additional vertices, with both of these vertices connected by edges to every vertex of the $k$-cycle.  The example of $\Bipyr  (7)$ can be seen in Figure~\ref{fig:seven-pack}, left.

\begin{figure}[t]
\centering
\vspace*{-2pt}
\includegraphics[scale=0.35]{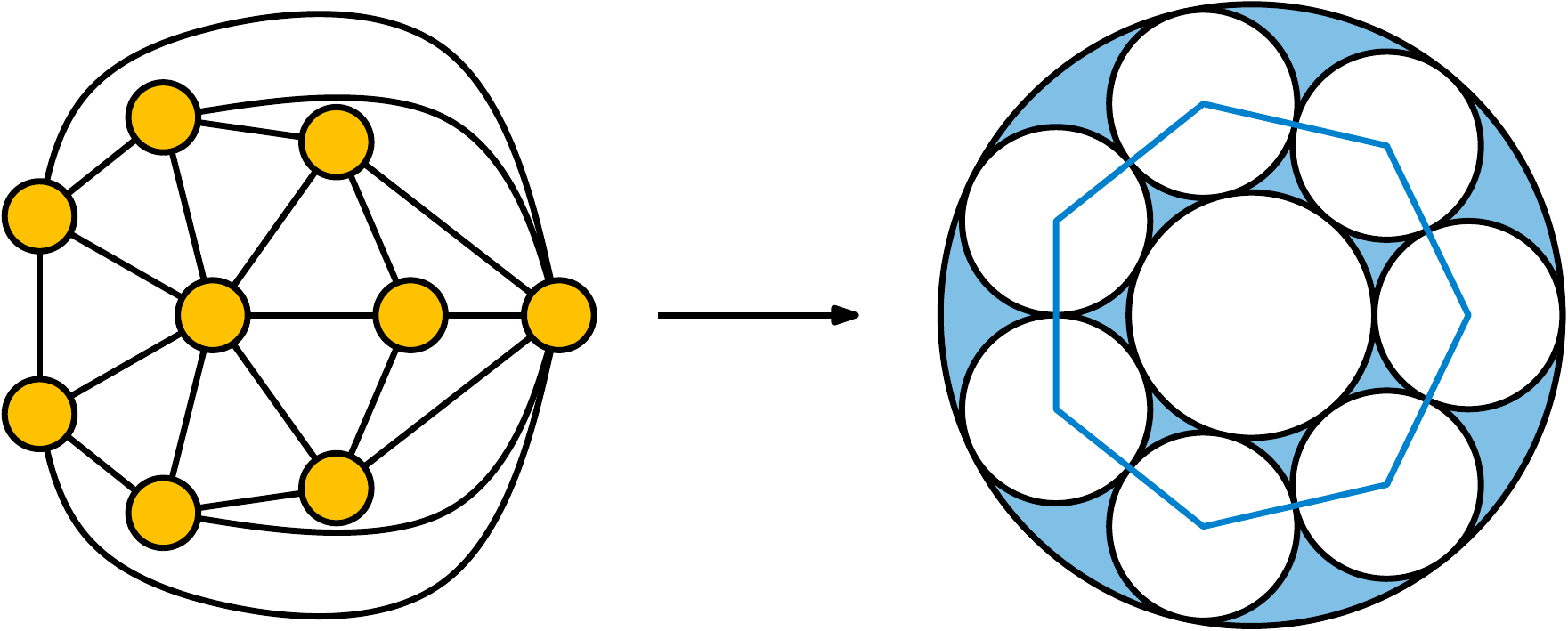}
\vspace*{-4pt}
\caption{\label{fig:seven-pack} The graph $\Bipyr (7)$ and its associated concentric circle packing.}
\end{figure}

\begin{theorem}
\label{thm:no-quadratic}
There exists a graph whose circle packings cannot be constructed by a quadratic computation tree.
\end{theorem}

\begin{proof}
Consider the circle packing of $\Bipyr (7)$ in which the two hubs are represented by concentric circles,
centered at the origin, with the other circle centers all on the unit circle and with one of them on the $x$ axis. One of the centers of this packing is at the root of unity $\zeta_7$. By Lemma~\ref{lem:roots-of-unity},
$ [\rationals (\zeta_7) : \rationals] = \phi(7) = 6$.
$6$ is not a power of two, so by Lemma~\ref{lem:extensions} $\zeta_7$ cannot be constructed by a quadratic computation tree. By Lemma~\ref{lem:concentric}, neither can any other packing for the same graph.
\end{proof}

In Appendix~\ref{app:circles}, we prove that certain circle packings also cannot be constructed by radical computation trees
nor by
bounded-degree root computation trees.

\section{Conclusion}
We have shown that several types of graph drawing cannot be constructed by models of computation that allow computation of arbitrary-degree radicals, nor by models that allow computation of the roots of bounded-degree polynomials. Whether the degree of these polynomials must grow linearly as a function of the input size, or only proportionally to a sublinear power, remains subject to an open number-theoretic conjecture.

It is natural to ask whether these drawings might be computable in a model of computation that allows both arbitrary-degree radicals and bounded-degree roots.
We leave this as open for future research.

\subsection*{Acknowledgements}

\ifFull
\else
We used the Sage software package
\fi
to perform preliminary calculations of the Galois groups of many drawings. 
Additionally, we thank Ricky Demer on MathOverflow for
\ifFull
\else
guiding us to research on large factors of $\phi(n)$.
\fi

{\raggedright
\ifFull
\else
\bibliographystyle{splncs}
\fi
\bibliography{paper}}

\clearpage
\begin{appendix}
\section{A Brief Review of Algebraic Graph Drawing}
\label{app:gd}
In this appendix, we provide a brief review of algebraic graph drawing, for
readers unfamiliar with these topics.

\subsection{Force-Directed Graph Drawing}
Force-directed algorithms are among the most popular and flexible general purpose graph drawing algorithms. They work by setting up a system of forces between vertices in the graph and then performing an iterative algorithm to attempt to reach an equilibrium state. By choosing an appropriate balance of forces, these algorithms can readily produce aesthetically pleasing drawings that exhibit the structure of the graph being drawn. More details on such algorithms can be found in several surveys on the subject~\citeappendix{GD-Handbook-13,BatEadTam-98,KauWag-01}.

In this paper, we focus our attention on the two most popular force-based drawing algorithms. We consider the Fruchterman and Reingold algorithm, which views the vertices as repelling charged particles connected by springs of rest length zero, and the Kamada and Kawai algorithm, which views graphs as a system is which every pair of vertices is connected by a spring whose spring constant and rest length is based on the graph theoretic distance between the vertices. The Fruchterman and Reingold algorithm computes a local minimum by simulating the motion induced by the forces, i.e., at each step the vertices are moved in a direction based on the current total force at the vertex. On the other hand, the Kamada and Kawai algorithm defines a total energy function for the system and then attempts to minimize this function by moving one vertex at a time, using a two-dimensional Newton--Raphson method. In both algorithms, computing an equilibrium may be viewed as solving a large system of polynomial equations in many variables.~\cite{FruRei-SPE-1991,KamKaw-IPL-1989}

A straightforward implementation of a force directed drawing algorithm would require $\Omega(n^2+m)$ work per iteration, as the pairwise forces must be computed between every pair of vertices in addition to the edge forces. This slow runtime would limit the size of graphs on which this method can be used. Researchers have found that by using the multipole method of $n$-body simulation \citeappendix{BarHut-Nat-1986,Greengard-1988} the work per iteration can be reduced to $O(n\log n)$. These fast algorithms, in combination with the parallelism of modern GPUs \citeappendix{GodHobJar-GD-2009,HacJun-GD-2005} allow force directed algorithms to be run on graphs with a hundred thousand nodes in under ten seconds.

\subsection{Spectral Graph Drawing}
Another family of general purpose graph drawing method are the spectral methods. 
To produce a spectral graph drawing of a graph $G$, we define an 
associated matrix, $M$, and, from the eigenvectors,
$u_1, u_2, u_3, \cdots, u_n$ of $M$ (ordered by eigenvalues), we choose two vectors $u_r$ and $u_s$. The coordinates in $\reals^2$ of a vertex $i$ in the drawing of $G$ are given by $(u_r[i], u_s[i])$. The choice of $M$, $r$, and $s$ determine the aesthetics of the drawing, and can be motivated by viewing the eigenvectors as solutions to optimization problems~\cite{Kor-CMA-2005}.

In 1970, Hall was the first to propose such a method for graph drawing, using the Laplacian matrix and the eigenvectors $u_2$ and $u_3$~\citeappendix{Hal-MS-70}. Later Manolopoulos and Fowler used the adjacency matrix to draw molecular graphs with the eigenvectors chosen based on the molecule being drawn~\citeappendix{ManFow-JCP-92}. Brandes and Willhalm used the $u_2$ and $u_3$ eigenvectors of a \emph{relaxed Laplacian}, $\rlaplacian(G) = \laplacian(G) - \rho\degree(G)$~\citeappendix{BraWil-SDV-02}. More recently, Koren has used the transition matrix and the eigenvectors $u_{n-2}$ and $u_{n-1}$~\cite{Kor-CMA-2005}.

Fast iterative algorithms for the numerical computation of the eigenvectors useful for graph drawing have been developed, which make spectral graph drawing practical for graphs with tens of millions of vertices and edges \cite{Kor-CMA-2005} \citeappendix{Puppe08, KorCarHar-IV-2002}. Unlike many force directed methods, these methods can be guaranteed to converge to a unique solution, rather than getting stuck in local optima. Based on this property, it has been claimed that in spectral drawing an exact solution may be computed, as opposed to typical NP-hard graph drawing formulations~\cite{Kor-CMA-2005}. In light of the results in this paper, we feel that it is more appropriate to say that there exists a single solution which may be efficiently approximated. 

\subsection{Multidimensional Scaling}

Multidimensional scaling algorithms are a family of graph drawing methods originally presented by Kruskal and Seery~\cite{KruSee-CSG-80}.  These techniques attempt to place vertices into $d$-dimensional space (where $d$ is usually two or three) such that the geometric distance of each pair of vertices is approximately close to some measure of their graph-theoretic distance. 
It can be seen as closely related both to Kamada--Kawai drawing (which attempts to fit the drawing to target distances between vertices) and spectral drawing (which uses matrices and their eigenvectors to construct drawings).

We encapsulate the target measure of distance in a matrix $D$ of pairwise squared vertex distances in $G$.  Formally we wish to find a matrix of vertex positions, $X = [x_1,\dots,x_n]^T$, so that for each pair of vertices $i$ and $j$, $||x_i - x_j||^2 \approx D_{i,j}$.  When the matrix $D$ is the matrix of the squared graph theoretic distances for all vertex pairs, this is classical multidimensional scaling~\citeappendix{BraPic-GD-07}.  In this case, to compute $X$, we first compute a derived matrix $B$ by \emph{double centering} the matrix $D$; this is an operation that combines $D$ linearly with its row averages, column averages, and overall average, in a way that (if $D$ were truly a matrix of squared Euclidean distances) would produce a matrix of dot products  $B= XX^T$. Then, to recover the matrix $X$ of vertex positions, we factor $B$ as $B=V\Lambda V^T$ where $V$ is the orthonormal matrix of the eigenvectors of $B$, sorted by eigenvalue, and $\Lambda$ is the matrix of corresponding eigenvalues of $B$.  The final vertex positions are obtained by taking the first $d$ columns of the matrix $V\Lambda^{1/2}$.

The computation of the matrix $D$ of squared graph theoretic distances, by an all pairs shortest path algorithm, and the computation of the eigenvectors of~$B$, by repeated multiplication and orthogonalization, are polynomial but can be expensive for large graphs.  Alternative variants of multidimensional scaling that reduce its computation time at the cost of solution quality have been proposed~\citeappendix{BraPic-GD-07}.  These other variants, including landmark and pivot multidimensional scaling, construct $D$ as the squared distance matrix of a smaller sample of vertex pairs and use slightly different techniques to compute $X$. These methods only approximate the $X$ from classical multidimensional scaling, but can be feasibly run on much larger graphs.

Multidimensional scaling can also be computed by a technique called majorization \citeappendix{GanKorNor-GD-05}.  The energy or stress of a drawing is defined to be a weighted sum of squared differences between vertex pair distances in the graph and their real positions.  A solution can be computed iteratively by bounding the energy from above with a convex function, based on the current solution, and setting each successive solution to be the minimizer of this convex function.

\subsection{Circle Packings}
The famous circle packing theorem of Koebe, Andreev, and Thurston states that every planar graph can be represented by a collection of interior-disjoint circles in the Euclidean plane, so that each vertex of the graph is represented by one circle and each edge is represented by a tangency between two circles~\cite{Koe-BSAWL-36} \citeappendix{Ste-ICP-05}. If the given graph is a maximal planar graph, the circle packing is unique up to M\"obius transformations; it can be made completely unique by packing the circles on a sphere rather than on the plane and by choosing a M\"obius transformation that maximizes the minimum radius of the circles~\citeappendix{BerEpp-WADS-01}.
Circle packings have many algorithmic applications, detailed below.

Efficient numerical algorithms for computing a circle packing representing a given graph are known, in time polynomial in the number of circles and the desired numerical precision~\citeappendix{ColSte-CGTA-03,Moh-DM-93,Moh-EJC-97}. These algorithms work with a system of radii of the circles, leaving their geometric placement for later. Starting from an inaccurate initial system of radii, they repeatedly improve this system by choosing one of the circles and replacing its radius by a new number that would allow the circle to be precisely surrounded by a ring of circles with its neighbors' radii. Each such replacement can be performed by a simple calculation using trigonometric functions, and the system of radii rapidly converges to values corresponding to a valid circle packing. Once the radii have been accurately approximated, the locations of the circle centers of the circles can be calculated by a process of triangulation. These algorithms have been implemented by multiple researchers---our figures are the output of a Python implementation initially developed for a graph drawing application~\citeappendix{Epp-GD-12}---and they work well in practice.

Although the known algorithms for circle packing use trigonometry, the circle packings themselves are algebraic: it is straightforward to write out a system of quadratic equations for variables describing the centers and radii of the circles, with each equation constraining two circles to be tangent.  The real-valued solutions to these equations necessarily include the desired circle packings, although they may also include other configurations of circles that have the proper tangencies but are not interior-disjoint.

Nevertheless, despite being an algebraic problem with many applications and with algorithms that are efficient in practice, we do not know of a strongly polynomial algorithm for circle packing, one that computes the solution exactly rather than numerically, and uses a number of computational steps that depends polynomially on the size of the input graphs but does not depend on the desired numerical precision of the output. The absence of such an algorithm cannot be explained solely by the high degree of the polynomials describing the solution, because the system of polynomials for a circle packing only has degree two, and because there are other problems (such as the construction of regular $2^n$-gons) that  have high degree and yet are easily solvable (for instance as an explicit formula or by compass and straightedge). In this paper, we use more subtle properties of the polynomials describing circle packings, based on Galois theory, to explain why an efficient exact circle packing algorithm does not exist.

\Emph{Applications of circle packing.}
The circle packing theorem, and algorithms based on it for transforming arbitrary planar graphs into tangent circle representations, have become a standard tool in graph drawing. Graph drawing results proved using circle packing include the fact that every planar graph of bounded degree can be drawn without crossings with a constant lower bound on its angular resolution (the minimum angle between incident edges)~\citeappendix{MalPap-SJDM-94} and with edges that have a constant number of distinct slopes~\citeappendix{KesPacPal-GD-10}. Circle packings have also been used to draw graphs on the hyperbolic plane~\citeappendix{Moh-GD-99}, to draw planar graphs on spheres in a way that realizes all of the symmetries of the planar embedding~\citeappendix{BerEpp-WADS-01}, to construct convex polyhedra that represent planar graphs~\citeappendix{Rot-GD-11},  to find \emph{Lombardi drawings} of planar graphs of degree at most three, drawings in which the edges are represented as circular arcs that surround each vertex by angles with equal areas~\citeappendix{Epp-GD-12}, to represent 4-regular planar graphs as the arcs and intersection points of an arrangement of circles that may cross each other~\citeappendix{BekRaf-GD-12}, to construct drawings in which each vertex is incident to a large angle~\citeappendix{AicRotSch-CGTA-12}, and to construct \emph{confluent drawings}, drawings in which edges are represented as smooth paths through a system of tracks and junctions~\citeappendix{EppHolLof-GD-13}.

Beyond graph drawing, additional applications of circle packing include algorithmic versions of the Riemann mapping theorem on the existence of conformal maps between planar domains~\citeappendix{Ste-CMFT-97}, unfolding human brain surfaces onto a plane for more convenient visualization of their structures~\citeappendix{HurBowSte-MICCAI-99},
finding planar separators~\citeappendix{MilTenThuVav-JACM-97,EppMilTen-FI-95},
approximation of dessins d'enfant (a type of graph embedding used in algebraic geometry)~\citeappendix{BowSte-MAMS-04},
and the geometric realization of soap bubbles from their combinatorial structure~\citeappendix{Epp-SCG-13}.

\clearpage
\section{Omitted Proofs of Lemmas and Theorems}
\label{app:proofs}
In this appendix, we provide details for omitted proofs of lemmas used in our
paper.

\subsection{Proof of Lemma~\ref{lem:mob-trans}}

Recall that Lemma~\ref{lem:mob-trans} states that, given any two disjoint circles, a M\"{o}bius transformation mapping them to two concentric circles can be constructed using a quadratic computation tree.

\begin{proof}
Such a transformation can be achieved by an inversion centered at one of the two \emph{limiting points} of the two circles. These two points lie on the line connecting the circle centers, at equal distances from the point $x$ where the radical axis of the two circles (the bisector of their power diagram) crosses this line. The distance from $x$ to the limiting points equals the power distance from $x$ to the two circles (the length of a tangent line segment from $x$ to either circle)~\citeappendix{Joh-CAGD-93}. From these facts it is straightforward to compute a limiting point, and hence the transformation, using only arithmetic and square root operations. Weisstein~\citeappendix{Mathworld-LimitingPoint} provides an explicit formula.
\end{proof}

\subsection{Proof of Lemma~\ref{lem:extensions}}
Recall that Lemma~\ref{lem:extensions} states,
if $\alpha$ can be computed by a root computation tree of degree $f(n)$, then $[\rationals(\alpha):\rationals]$ is $f(n)$-smooth, i.e., none of its prime factors are greater than $f(n)$. In particular, if $\alpha$ can be computed by a quadratic computation tree, then $[\rationals(\alpha):\rationals]$ is a power of two.

\begin{proof}
Annotate each node of the given root computation tree with a minimal extension of the rational number field containing all of the values computed along the path to that node. This field is an extension of  the field for the parent node in the tree by the root of a polynomial of degree at most $f(n)$, so as a field extension it has degree at most $f(n)$ (Proposition~4.3.4 of \cite{Cox2012}, p.~89).
Therefore, the field for each node can be constructed by a sequence of extensions of the rational numbers, each of degree at most $f(n)$. Since $\rationals(\alpha)$ is a subfield of this field, it can also be constructed in the same way. The degree of a sequence of extensions is the product of the degrees of each extension (the ``tower theorem'', Theorem~4.3.8 of \cite{Cox2012}, p.~91). Since each of these extensions is $f(n)$-smooth, so is their product.
\end{proof}

\subsection{Proof of Theorem~\ref{thm:kk-radical}}
Recall Theorem~\ref{thm:kk-radical}, which states that
there exists a graph on four vertices such that it is not possible on a radical computation tree to construct the coordinates of every possible Kamada and Kawai equilibrium.

\begin{proof}
With respect to the four-vertex graph of Figure~\ref{fig:kk-radical},
we define the following variables:
\begin{align*}
a &=\text{the distance from $u_0$ to $u_1$}\\
b &=\text{the horizontal distance from $u_1$ to $u_2$}\\
c &=\text{the vertical distance from $u_1$ to $u_2$}\\
d &=\text{the distance from $u_1$ to $u_2$}\\
e &=\text{the distance from $u_0$ to $u_2$}
\end{align*}
We make some assumptions on the positions of the vertices. We assume that the line defined by the positions of $u_0$ and $u_1$ meet the line defined by the positions of $u_2$ and $u_3$ at a right angle, the vertices are ordered as in the figure, and that there are no vertex or edge overlaps. These assumption correspond to the equilibrium that is typically produced by the Kamada and Kawai algorithm.
With these variables, the local optimum conditions for Kamada and Kawai are
as follows:
\begin{align*}
\frac{\partial E}{\partial x_0} &= \frac{-3/2ae + a - 1/2be + b + e}{e} =0\\
\frac{\partial E}{\partial y_0} &= 0\\
\frac{\partial E}{\partial x_1} &= \frac{ad - 2bd + 2b - d}{d} =0\\
\frac{\partial E}{\partial y_1} &= 0\\
\frac{\partial E}{\partial x_2} &= \frac{1/4ade - 1/2ad + 5/4bde - 1/2bd - be}{de} =0\\
\frac{\partial E}{\partial y_2} &= \frac{13/4cde - 1/2cd - ce - de}{de} =0\\
\frac{\partial E}{\partial x_3} &= \frac{1/4ade - 1/2ad + 5/4bde - 1/2bd - be}{de} =0\\
\frac{\partial E}{\partial y_3} &= \frac{-13/4cde + 1/2cd + ce + de}{de} = 0.
\end{align*}
We also have the additional constraints,
\[
d^2 - b^2 - c^2 = 0 \quad\quad\text{and}\quad\quad e^2 - (a+b)^2 - c^2,
\]
which follow from our choice of variables. From this system of equations, we can, with the aid of a computer algebra system, compute the Groebner Basis of the system and extract a single polynomial, denoted $p$, that $c$ must satisfy:
\begin{eqnarray*}
p(c) &=& 365580800000000c^{18}
- 2065812736000000c^{17}
+ 5257074184960000c^{16}\\
&&- 7950536566252800c^{15}
+ 7939897360159392c^{14}
- 5501379135910008c^{13}\\
&&+ 2703932242407045c^{12}
- 947252378063088c^{11}
+ 234371204926092c^{10}\\
&&- 40028929618536c^9
+ 4535144373717c^8
- 317453745456c^7\\
&&+ 11493047016c^6
- 83177280c^5
+ 167184c^4.
\end{eqnarray*}
The polynomial $p$ factors into $c^4$ and an irreducible factor of degree $14$. The $c^4$ factor corresponds to degenerate drawings with $c=0$ and with $u_2$ and $u_3$ drawn at the same point of the plane; since the expected drawing is of a different type, we can ignore this factor.
Let $f(c)$ be the factor of degree $14$. We can convert $f$ to a monic polynomial using the same techniques as before, producing the following polynomial:
\begin{eqnarray*}
g(x) &=& x^{14}f(258/x)/167184\\
&=& x^{14} - 128360x^{13}\\
&&+ 4575935386x^{12}\\
&&- 32609554186008x^{11}\\
&&+ 120191907907039173x^{10}\\
&&- 273701889217560990672x^9\\
&&+ 413454551042624579937072x^8\\
&&- 431130685015107552530542464x^7\\
&&+ 317510974076480215971285088080x^6\\
&&- 166668765204034179394613907054336x^5\\
&&+ 62060780922813932272692806330099712x^4\\
&&- 16033136614269762618278694793639526400x^3\\
&&+ 2735179704826314422602131722817699840000x^2\\
&&- 277301626082465808611849917345431552000000x\\
&&+ 12660899181603462048518168020372684800000000.
\end{eqnarray*}
The polynomial $g$ can be algorithmically verified to be irreducible via a computer algebra system. Its discriminant is
\begin{multline*}
2^{188} \cdot 3^{90} \cdot 5^{26} \cdot 7^{25} \cdot 13^{12} \cdot 43^{156} \cdot 987004987^2 \cdot 142654761797^2\\ \cdot 20040994351453^2 \cdot 1908270249053041126780511430661^2\\ \cdot 2068784364712376186850628387585613,
\end{multline*}
and we have the following factorizations of $g$ into irreducible polynomials modulo the primes $67$ and $113$, which do not divide the discriminant.
\begin{align*}
g(x) &\equiv (x + 25)  (x^{13} + 54x^{12} + 62x^{11} + 40x^{10} + 48x^9 + 52x^8 + 10x^7\\
&+ 38x^6 + 24x^5 + 14x^4 + 30x^3 + 17x^2 + 65x + 34)
&\hspace{-2em}\pmod{67}\\
g(x) &\equiv (x + 50)  (x^2 + 15x + 49)  (x^{11} + 56x^{10} + 15x^9 + 94x^8 + 60x^7\\
&+ 61x^6 + 13x^5 + 103x^4 + 53x^3 + 11x^2 + 6x + 13)
&\hspace{-2em}\pmod{113}
\end{align*}
By Dedekind's theorem, the factorization modulo $67$ implies the existence of a $13$-cycle in $\Gal(\splitting(g)/\rationals)$, and the factorization modulo $113$ implies the existence of a permutation that is the composition of a transposition and a $11$-cycle. The second permutation produces a transposition when raised to the eleventh power. Now, by Lemma~\ref{lem:swap+cycle} $\Gal(\splitting(g)/\rationals) = S_{14}$. So by Lemma~\ref{lem:radical-tree} the value of $c$ cannot be computed by a radical computation tree.
\end{proof}

\clearpage
\section{Additional Impossibility Results for Spectral Graph Drawing}
\label{app:spectral}
In this appendix, we provide additional impossibility results for spectral
graph drawing, based on the 12-vertex graph, $H$, shown in
Figure~\ref{fig:h-graph}.

\begin{figure}[hbt]
\centering
\includegraphics[scale=0.7]{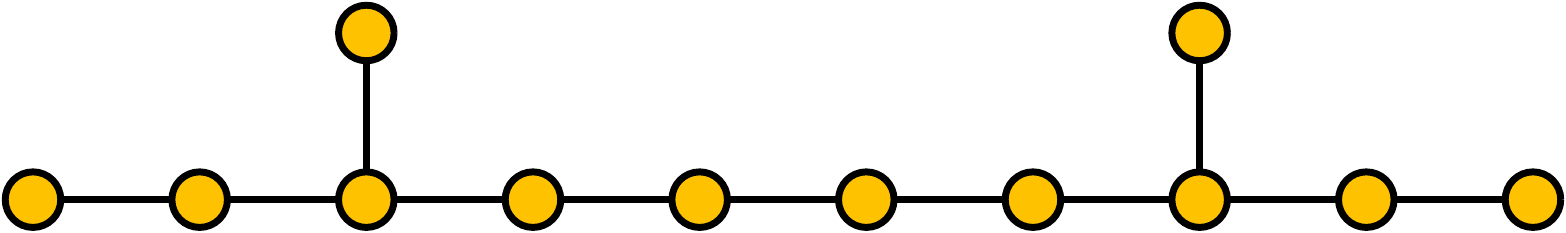}
\caption{A graph, $H$, whose adjacency and transition eigenvectors are uncomputable in a radical tree.}
\label{fig:h-graph}
\end{figure}

The adjacency matrix of $H$ is given by
\[
\adjacency(H) =
\begin{pmatrix}
 & 1&  &  &  &  &  &  &  &  &  &  \\
1&  & 1&  &  &  &  &  &  &  &  &  \\
 & 1&  & 1&  &  &  &  &  &  & 1&  \\
 &  & 1&  & 1&  &  &  &  &  &  &  \\
 &  &  & 1&  & 1&  &  &  &  &  &  \\
 &  &  &  & 1&  & 1&  &  &  &  &  \\
 &  &  &  &  & 1&  & 1&  &  &  &  \\
 &  &  &  &  &  & 1&  & 1&  &  & 1\\
 &  &  &  &  &  &  & 1&  & 1&  &  \\
 &  &  &  &  &  &  &  & 1&  &  &  \\
 &  & 1&  &  &  &  &  &  &  &  &  \\
 &  &  &  &  &  &  & 1&  &  &  &  \\
\end{pmatrix}
\]
and its characteristic polynomial can be computed to be
\begin{multline*}
q(x) = \characteristic(\adjacency(H))\\ = (x^6 - x^5 - 5 x^4 + 4 x^3 + 5 x^2 - 2 x - 1)\\   (x^6 + x^5 - 5 x^4 - 4 x^3 + 5 x^2 + 2 x - 1).
\end{multline*}
Let $q_0$ and $q_1$ be the factors of $q$ in the order given above. First, observe that $q_0(x) = q_1(-x)$, which implies that we need only compute the Galois group of $q_0$. The polynomial $q_0$ is irreducible by Lemma~\ref{lem:2n+1}, as it produces a prime for 13 integer inputs in the range from 0 to $50$. The discriminant of $q_0$ is $592661$ (a prime), and we have the following factorizations of $q_0$ into irreducible polynomials modulo the primes $13$ and $7$, which do not divide the discriminant:
\begin{align*}
q_0(x) &\equiv  (x + 9)  (x^5 + 3x^4 + 7x^3 + 6x^2 + 3x + 10) &\pmod{13}\\
q_0(x) &\equiv  (x + 2)  (x^2 + 5x + 5)  (x^3 + 6x^2 + x + 2) &\pmod{7}.
\end{align*}
By Dedekind's theorem, the factorization modulo 13 implies the existence of a $5$-cycle in $\Gal(\splitting(q_0)/\rationals)$ and factorization modulo $7$ implies the existence of a permutation that is the composition of a transposition and a $3$-cycle in $\Gal(\splitting(q_0)/\rationals)$. The second permutation when cubed yields a transposition. Therefore, Lemma~\ref{lem:swap+cycle} implies $\Gal(\splitting(q_0)/\rationals) = S_8$. So by Lemma~\ref{lem:radical-tree} the eigenvalues of $\adjacency(H)$ are not computable in a radical computation tree.

\begin{theorem}
There exists a graph on 12 vertices such that it is not possible to construct a spectral graph drawing based on the adjacency matrix in a radical computation tree.
\end{theorem}

The transition matrix of $H$ is given by
\[
\transition(H)=
\begin{pmatrix}
   &   1&    &    &    &    &    &    &    &    &    &    \\
1/2&    & 1/2&    &    &    &    &    &    &    &    &    \\
   & 1/3&    & 1/3&    &    &    &    &    &    & 1/3&    \\
   &    & 1/2&    & 1/2&    &    &    &    &    &    &    \\
   &    &    & 1/2&    & 1/2&    &    &    &    &    &    \\
   &    &    &    & 1/2&    & 1/2&    &    &    &    &    \\
   &    &    &    &    & 1/2&    & 1/2&    &    &    &    \\
   &    &    &    &    &    & 1/3&    & 1/3&    &    & 1/3\\
   &    &    &    &    &    &    & 1/2&    & 1/2&    &    \\
   &    &    &    &    &    &    &    &   1&    &    &    \\
   &    &   1&    &    &    &    &    &    &    &    &    \\
   &    &    &    &    &    &    &   1&    &    &    &    \\
\end{pmatrix}
\]
and its characteristic polynomial can be computed to be
\begin{multline*}
r(x) = \characteristic(\transition(H))\\ = (x - 1)   (x + 1)   (x^5 - 1/2 x^4 - 11/12 x^3 + 1/3 x^2 + 1/6 x - 1/24)\\   (x^5 + 1/2 x^4 - 11/12 x^3 - 1/3 x^2 + 1/6 x + 1/24)
\end{multline*}
Let $r_1$, $r_2$, $r_3$ and $r_4$ be the factors of $r$ in the order given above.
As before, we have a relation between $r_2$ and $r_3$, $r_2(x) = -r_3(-x)$, which means that we need only compute the Galois group for $r_2$. First, we covert $r_2$ into a monic polynomial with integer coefficients,
\[
s(x) = -24x^5r_2(1/x) = x^5 - 4x^4 - 8x^3 + 22x^2 +12x - 24.
\]
The discriminant of $s$ is $2^8\cdot 3\cdot 97 \cdot 6947$, and we have the following factorization of $s$ into irreducible polynomials modulo the primes $11$ and $5$, which do not divide the discriminant:
\begin{align*}
s(x) &\equiv (x + 1)  (x^4 + 6x^3 + 8x^2 + 3x + 9) &\pmod{11}\\
s(x) &\equiv (x^2 + x + 1)  (x^3 + x + 1) &\pmod{5}
\end{align*}
By Dedekind's theorem, the factorization modulo $11$ implies the existence of a $4$-cycle in $\Gal(\splitting(s)/\rationals)$ and the factorization modulo $5$ implies the existence of a permutation that is the composition of a transposition and a $3$-cycle in $\Gal(\splitting(s)/\rationals)$. When cubed the second permutation produces a transposition. Therefore, Lemma~\ref{lem:swap+cycle} implies $\Gal(\splitting(s)/\rationals) = S_5$. So by Lemma~\ref{lem:radical-tree} the only eigenvalues of $\transition(G)$ that are computation in a radical computation tree are $1$ and $-1$. Since the roots of $r$ are in the interval $[-1,1]$, the only computable eigenvectors correspond to the largest and smallest eigenvalues, whose eigenvectors are given below.
\begin{align*}
u_1 &= (1, -1, 1, -1, 1, -1, 1, -1, 1, -1, -1, 1) \\
u_{12} &= (1, 1, 1, 1, 1, 1, 1, 1, 1, 1, 1, 1)
\end{align*}

\begin{theorem}
There exists a graph on twelve vertices such that the only spectral drawing based on the transition matrix that is computable in a radical computation tree uses the largest and smallest eigenvectors, and produces a drawing in which many vertices have coinciding positions.
\end{theorem}

\clearpage

\section{Additional Impossibility Results for Multidimensional Scaling}
\label{app:mds}

In this appendix, we provide additional impossibility results for the classical multidimensional scaling method.

\begin{figure}[hbt]
\centering
\includegraphics[scale=0.7]{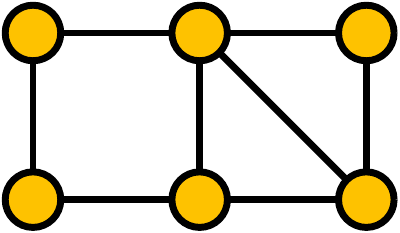}
\caption{A graph, $H$, whose classical multidimensional scaling coordinates are uncomputable by a radical tree.}
\label{fig:2-by-3-graph}
\end{figure}

The squared distance matrix for the graph in Figure~\ref{fig:2-by-3-graph} is 

\[
\begin{pmatrix}
0 & 1 & 1 & 4 & 4 & 9\\
1 & 0 & 1 & 1 & 1 & 4\\
1 & 1 & 0 & 1 & 4 & 4\\
4 & 1 & 1 & 0 & 4 & 1\\
4 & 1 & 4 & 4 & 0 & 1\\
9 & 4 & 4 & 1 & 1 & 0
\end{pmatrix}.
\]

After double centering it becomes

\[
\begin{pmatrix}
-73/18  &  -11/9  & -31/18  &  23/18  &    7/9  &  89/18\\
 -11/9  &  -7/18  &    1/9  &    1/9  &  -7/18  &   16/9\\
-31/18  &    1/9  & -25/18  &  -7/18  &   19/9  &  23/18\\
 23/18  &    1/9  &  -7/18  & -25/18  &   19/9  & -31/18\\
   7/9  &  -7/18  &   19/9  &   19/9  & -43/18  &  -20/9\\
 89/18  &   16/9  &  23/18  & -31/18  &  -20/9  & -73/18
\end{pmatrix}
\]

and the characteristic polynomial of this double centered matrix is:

\[ p(x) = x(x^5 + 41/3x^4 + 19x^3 - 125x^2 - 88/3x^1 + 48). \]

Let $p_0$ and $p_1$ be the two irreducible factors of $p$ in the above order.  Because $p_1(2) < 0$ and $p_1(3) > 0$, the factor $p_1$ has a root greater than zero; therefore, $p_1$ is the factor with the largest root. This root is the largest eigenvalue.  We now convert $p_1$ into a monic polynomial with integer coefficients,

\[q(x) = x^5 - 88x^4 - 54000x^3 + 1181952x^2 + 122425344x + 1289945088.\]

The discriminant of $q$ is $2^{61} \cdot 3^{31} \cdot 12421 \cdot 3039011$.  Factoring $q$ into irreducible polynomials modulo $7$ and $11$ gives:

\begin{align*}
q(x) &\equiv (x + 4)  (x^4 + 7x^3 + 4x^2 + 8x + 9) &\pmod{11}\\
q(x) &\equiv (x^2 + 2x + 5)  (x^3 + x^2 + 5x + 1) &\pmod{7}
\end{align*}

Because neither $11$ nor $7$ divides the discriminant of $q$, Dedekind's theorem implies the existence of a $4$-cycle and the composition of a transposition with a $3$-cycle in $\Gal(\splitting(q)/\rationals)$. Taking the third power of the latter element gives an element that is just a transposition. The existence of these two elements in the Galois group implies, by Lemma~\ref{lem:swap+cycle}, that $\Gal(\splitting(q)/\rationals) = S_5$.

In multidimensional scaling the vertex positions are determined by multiplying a matrix of the first few eigenvectors by a matrix of the square roots of the corresponding eigenvalues.  Lemma~\ref{lem:radical-tree} implies that these eigenvalues cannot be computed in a radical computation tree, but we must still show that the vertex positions themselves also cannot be computed in this model. However, the columns of the matrix of vertex positions are themselves multiples of eigenvectors. If we could compute the vertex positions, we could use these eigenvectors to recover their corresponding eigenvalues. Since the eigenvalues cannot be computed, it follows that the vertex positions also cannot be computed.

\begin{theorem}
There exists a graph on six vertices such that the drawing produced by classical multidimensional scaling is not computable in a radical computation tree.
\end{theorem}

We leave as open for future research the problem of proving degree lower bounds for multidimensional scaling in the root computation tree. The technique that we used for the corresponding problem for other graph drawing techniques was to express the coordinates of drawings of highly-symmetric graphs using high-degree cyclotomic polynomials or (almost equivalently) high degree Chebyshev polynomials, but that does not seem to work in this case. For instance, the characteristic polynomial for the multidimensional scaling drawing of an $n$-vertex cycle graph cannot be a cyclotomic polynomial of order~$n$, because the characteristic polynomial for multidimensional scaling always includes zero as a root whereas the cyclotomic polynomial has all its roots nonzero.

\clearpage
\section{Additional Impossibility Results for Circle Packings}
\label{app:circles}
In this appendix, we provide additional impossibility results for circle packings.
\begin{theorem}
For arbitrarily large values of $n$,
there are graphs on $n$ vertices such that constructing a circle packing for the graph on a root computation tree requires degree $\Omega(n^{0.677})$. If there exist infinitely many Sophie Germain primes, then there are graphs for which constructing a circle packing requires degree $\Omega(n)$.
\end{theorem}

\begin{proof}
As in Theorem~\ref{thm:no-quadratic},
we consider packings of $\Bipyr (n-2)$ that have two concentric circles centered at the origin and all remaining circle centers on the unit circle. We choose $n=p+2$ where $p$ is a prime number for which $\phi(p)=p-1$ has a large prime factor $q$. If arbitrarily large Sophie Germain primes exist we let $q$ be such a prime and let $p=2q+1$. Otherwise, by Lemma~\ref{lem:large-prime-factor-of-phi} we choose $p$ in such a way that its largest prime factor $q$ is at least $p^{0.677}$.

By Lemma~\ref{lem:roots-of-unity}, we have:
\[ [\rationals (\zeta_{p}) : \rationals] = \phi(p) = p-1. \]
Thus, this extension is not $D$-smooth for any $D$ smaller than $q$, and every construction of it on a root computation tree requires degree at least $q$.
By Lemma~\ref{lem:concentric}, the same degree is necessary for constructing any packing of the same graph.
\end{proof}

\begin{figure}
\centering
\includegraphics[width=0.9\textwidth]{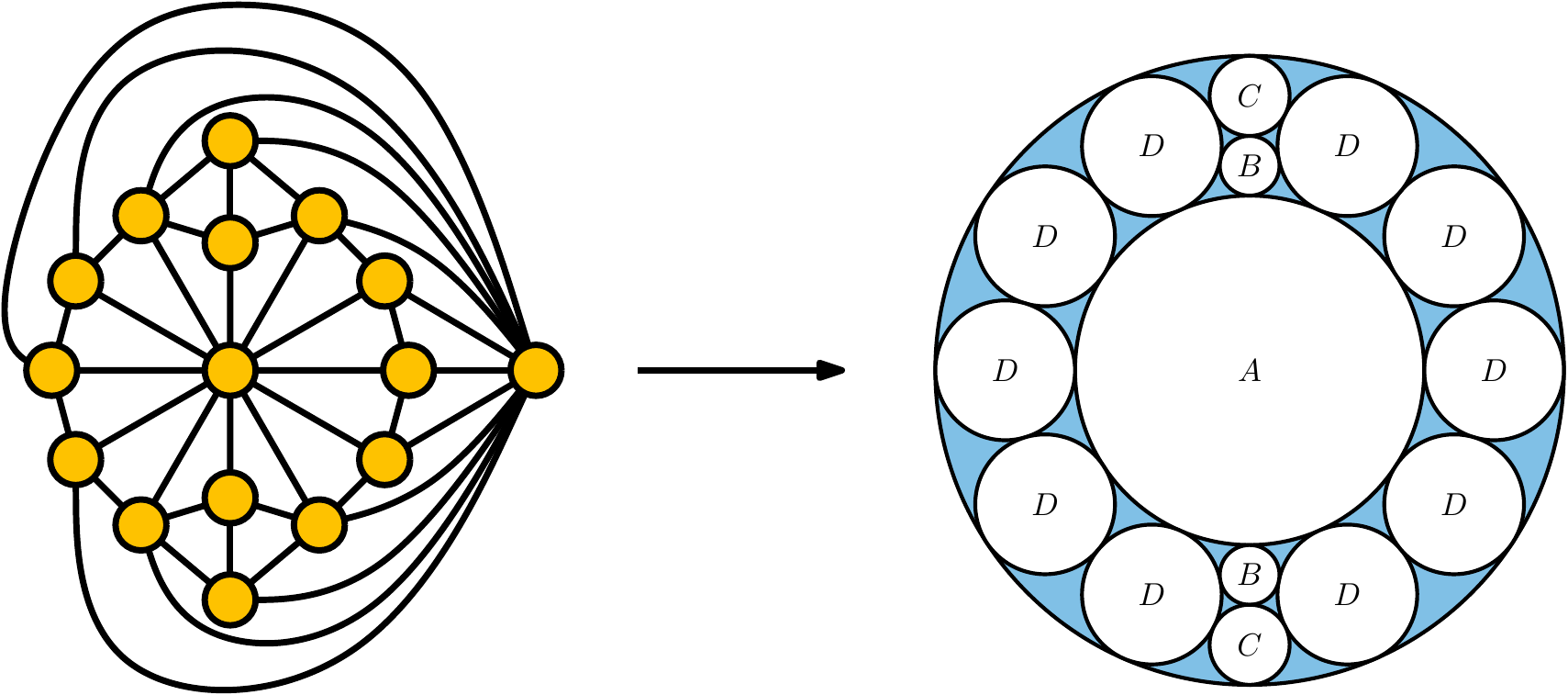}
\caption{An example of a graph (left) and its corresponding circle packing (right) where the circle packing is not constructible in a radical computation tree.}
\label{fig:two-five-graph}
\end{figure}

\Emph{Radical computation trees.} 
To show that circle packings are in general not constructible with a radical computation tree we consider the input graph shown in Figure~\ref{fig:two-five-graph},  together with a circle packing in which the circles corresponding to the two degree twelve vertices are concentric. We assume that this packing has been scaled so that the circles tangent to both concentric circles (the circles labeled with $D$) have radius equal to one, and so that (as in the figure) the smaller circle centers lie on the $y$ axis and two of the unit circle centers lie on the $x$ axis. The placement and radii of the circles in this packing may be determined from two values, namely the radius $a$ of the circle labeled $A$ and the radius $b$ of the circle $B$.

We use the following two simple trigonometric lemmas:

\begin{lemma}\label{lem:twin-circles}
If $\arccos(X) + \arccos(Y) = \pi$, then $X+Y=0$.
\end{lemma}

\begin{proof}
Let $X' = \arccos(X)$ and $Y' = \arccos(Y)$. Then we have
\begin{eqnarray*}
0 &=& \cos(X'+Y') - \cos(\pi) = \cos(X')\cos(Y')-\sin(X')\sin(Y') + 1 \\
  &=& (XY+1) - \sqrt{1-X^2}\sqrt{1-Y^2},
\end{eqnarray*}
which implies
\[
0 = (XY+1)^2 - (1-X^2)(1-Y^2) = X^2 + 2XY + Y^2 = (X+Y)^2.
\]
Thus, $X+Y = 0$.
\end{proof}

\begin{lemma}\label{lem:k5}
If $\arccos(U) + 2\arccos(V) = \pi/2$, then $4V^4 - 4V^2 + U^2 = 0$.
\end{lemma}

\begin{proof}
Let $U' = \arccos(U)$ and $V' = \arccos(V)$. Then we have
\begin{align*}
0 &= \cos(U' + 2V') - \cos(\pi/2)\\
&= \cos(U')(2\cos^2(V') - 1) - \sin(U')(2\sin(V')\cos(V'))\\
&= 2UV^2 - U - 2\sqrt{1-U^2}\sqrt{1-V^2}V
\end{align*}
which implies
\[
0 = (2UV^2 - U)^2 -4(1-U^2)(1-V^2)V^2 = 4V^4 - 4V^2 + U^2.
\]
Thus, $4V^4 - 4V^2 + U^2=0$.
\end{proof}

We now derive polynomial equations that these two radii must satisfy.
If we consider the triangle formed by the centers of circles $C$, $B$ and $D$, then the angle at the center of $B$ is given by $\arccos(X)$, where $X$ is given below. Similarly, if we consider the triangle formed by the centers of circles $B$, $D$ and $A$, then the angle at the center of $B$ is given by $\arccos(Y)$, where $Y$ is given below. These formulas follow from a direct application of the law of cosines.
\begin{align*}
X &= \frac{(b+(1-b))^2 + (b + 1)^2 - ((1-b)+1)^2}{2(b+(1-b))(b+1)}\\
Y &= \frac{(b+1)^2 + (b+a)^2 - (a+1)^2}{2(b+1)(b+a)}
\end{align*}
Now we have the relation $\arccos(X) + \arccos(Y) = \pi$, as the angles around $B$ sum to $2\pi$. This fact together with Lemma~\ref{lem:twin-circles} implies $a = 2b^2 / (1-2b)$. Thus, we can remove the variable $a$ from consideration as it can be computed from $b$ in a radical computation tree.

To find a polynomial with $b$ as its root we consider the angles around circle $A$. The angle at the center of circle $A$ in the triangle through the centers of the circle $A$, $B$ and $D$ is given by $\arccos(U)$, where $U$ is given below. Similarly, the angle at the center of circle $A$ in the triangle through the centers of $A$, $D$ and an adjacent $D$ is given by $\arccos(V)$, where $V$ is given below. Again, these formulas follow from the law of cosines.
\[
U = \frac{(a+b)^2 + (a+1)^2 - (b+1)^2}{2(a+b)(a+1)} = \frac{-6b^2 + 6b -1}{2b^2-2b+1}
\]
\[
V = \frac{(a+1)^2 + (a+1)^2 - (1+1)^2}{2 (a+1)(a+1)} = \frac{(2b^2 - 4b +1)(2b^2-1)}{(2b^2 - 2b +1)^2}
\]

Since the angles around the circle $A$ sum to $2\pi$ we have the relation $\arccos(U) + 2\arccos(V) = \pi/2$. Plugging the computed values of $U$ and $V$ into the formula of Lemma~\ref{lem:k5} yields the polynomial $f(b)$ as its numerator, where:
\begin{align*}
f(b) =\ &2304b^{16} - 18432b^{15} + 68096b^{14} - 154112b^{13} + 254720b^{12}\\ &- 363520b^{11} + 471424b^{10} - 501376b^9 + 390112b^8 - 208000b^7\\ &+ 73440b^6 - 17504b^5 + 3568b^4 - 896b^3 + 200b^2 - 24b + 1
\end{align*}

The polynomial $f(b)$ factors as the product of two irreducible eighth degree polynomials $f_0(b)$ and $f_1(b)$, below. We have the identity $f_0(b) = f_1(1-b)$ which appears to come from the symmetry between the outer and inner circles of the packing. For this reason the splitting field of $f(b)$ is equal to the splitting field of $f_0(b)$. Since the polynomial $f_0$ is not monic we will instead consider the monic polynomial $g(b) = b^8f_0(1/b)$ (this corresponds to reversing the order of the coefficients), which has the same splitting field:
\begin{align*}
f_0(b) &=48b^8 - 256b^7 + 592b^6 - 656b^5 + 336b^4 - 64b^3 + 4b^2 - 4b + 1\\
f_1(b) &= 48b^8 - 128b^7 + 144b^6 - 208b^5 + 336b^4 - 288b^3 + 116b^2 - 20b + 1\\
g(b) &= b^8 - 4b^7 + 4b^6 - 64b^5+336b^4-656b^3+592b^2-256b +48
\end{align*}

The polynomial $g(b)$ is irreducible by Lemma~\ref{lem:2n+1}, as it produces a prime for seventeen integer inputs in the range from $-119$ to $101$.
The discriminant of $g$ is $2^{52} \cdot 81637$, and we have the following factorization of $g$ into irreducible polynomials modulo the primes $3$ and $29$:
\begin{align*}
g(b) &\equiv b  (b^7 + 2b^6 + b^5 + 2b^4 + b^2 + b + 2) &&\pmod{3}\\
g(b) &\equiv (b + 9)  (b^2 + 23b + 12)  (b^5 + 22b^4 + 9b^3 + 20b + 23) &&\pmod{29}
\end{align*}
By Dedekind's theorem, the factorization modulo $3$ implies the existence of a $7$-cycle in $\Gal(\splitting(g)/\rationals)$, and the factorization modulo $29$ implies there is also a permutation that is the composition of a transposition and a $5$-cycle. Taking this second permutation to the fifth power eliminates the $5$-cycle, producing a transposition. By Lemma~\ref{lem:swap+cycle}, $\Gal(\splitting(g)/\rationals) = S_8$. So by Lemma~\ref{lem:radical-tree} the value of $b$ cannot be computed by a radical computation tree. Thus, by Lemma~\ref{lem:concentric}  a radical computation tree cannot construct any circle packing of the graph in Figure~\ref{fig:two-five-graph}, proving the following theorem.

\begin{theorem}
There exists a graph on sixteen vertices such that constructing a circle packing for the graph on a radical computation tree is not possible.
\end{theorem}

\Emph{Additional circle packings and their groups.}\label{app:gallery}
The Galois groups in this section were calculated using Sage. We identified the low-degree groups by using the PARI library, and the high-degree symmetric groups by using a brute force search for primes satisfying the conditions of Dedekind's theorem and Lemma~\ref{lem:swap+cycle}. To factor the polynomials arising in these computations we used the FLINT library.

\begin{figure}[t]
\centering
\includegraphics[height=1.5in]{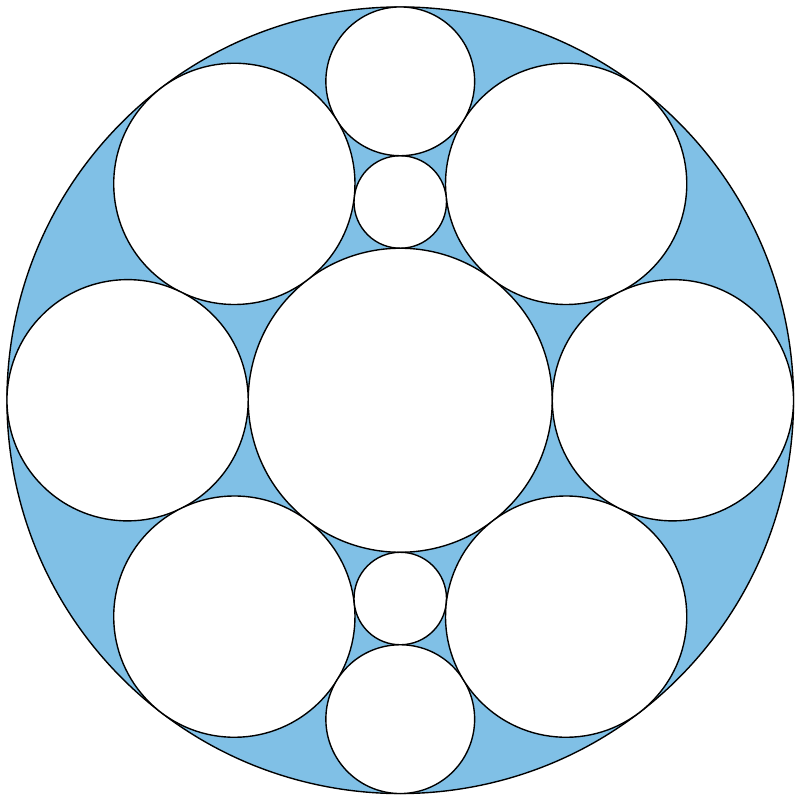}\hspace{1em}
\includegraphics[height=1.5in]{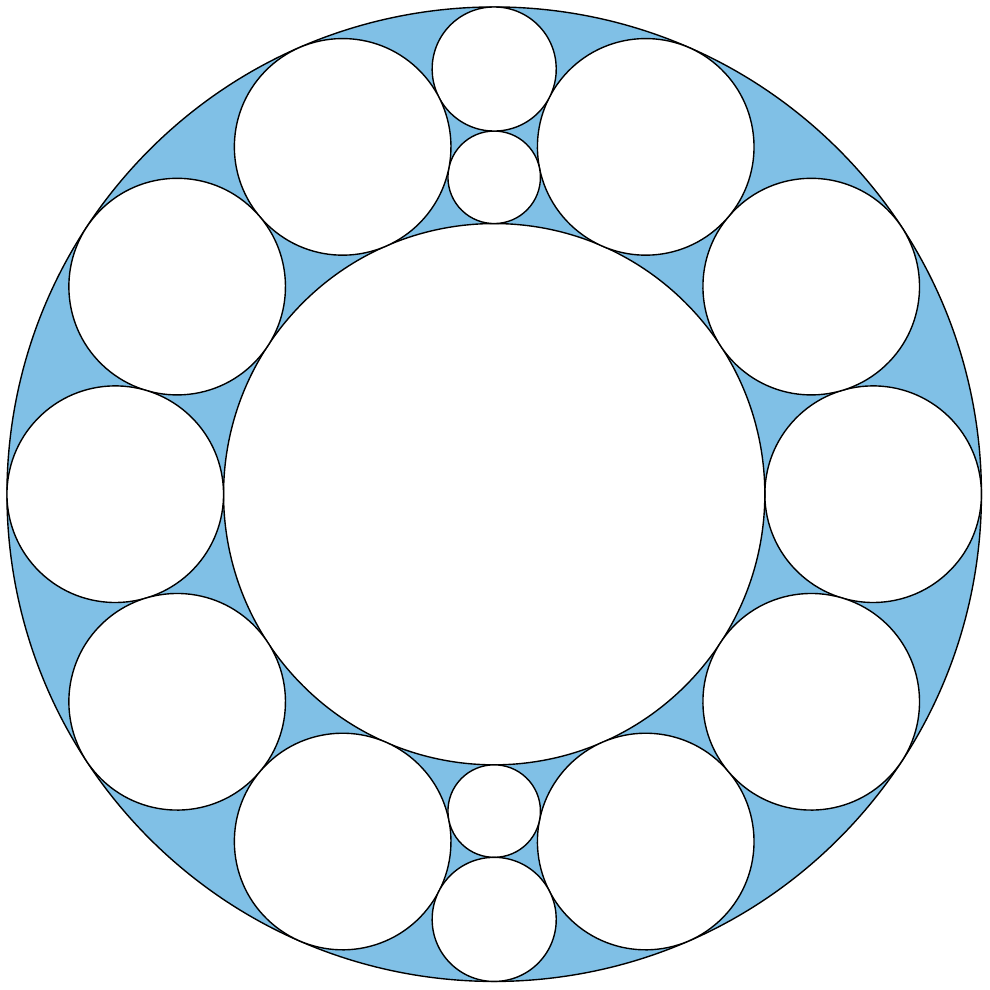}\\
\vspace{1em}
\includegraphics[height=1.5in]{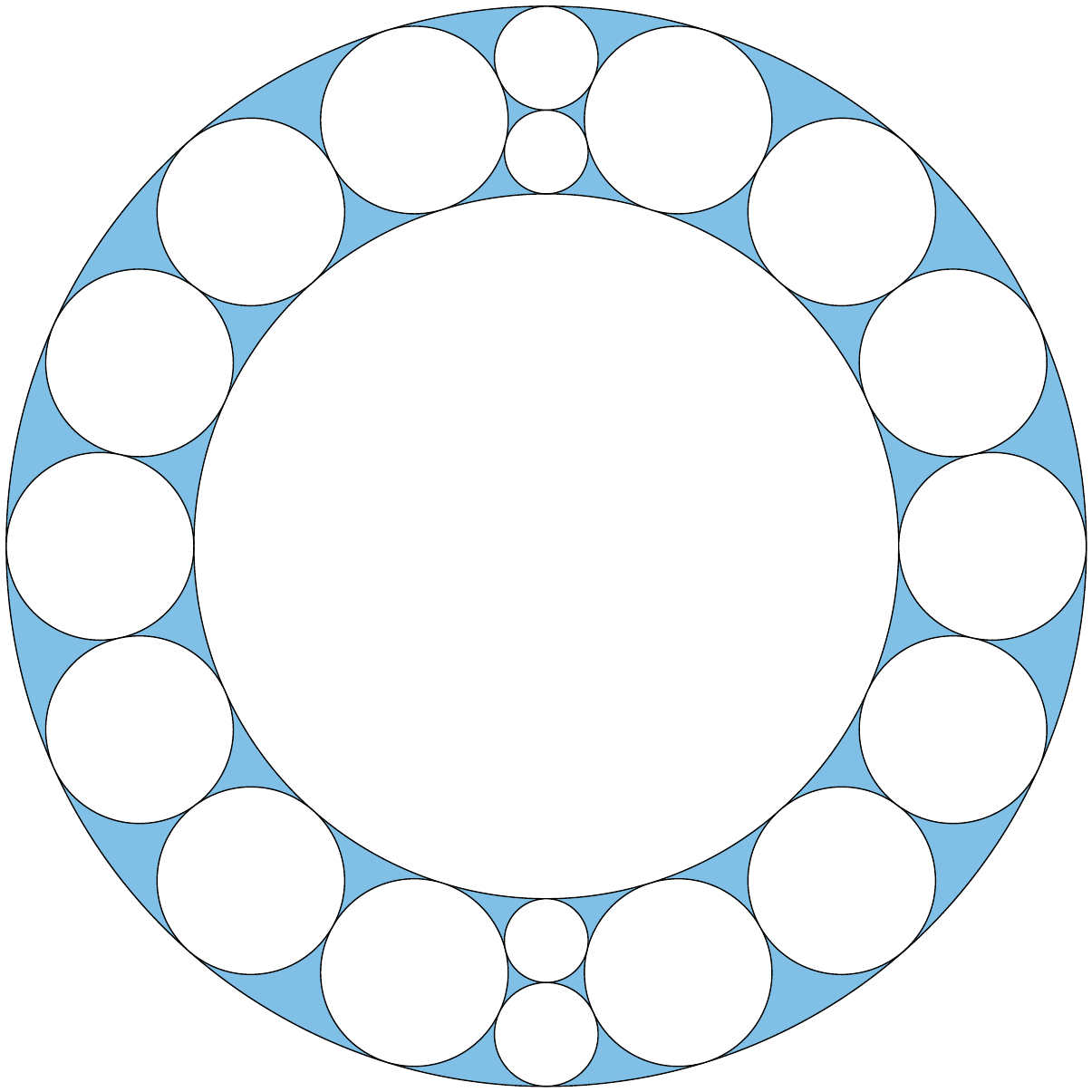}\hspace{1em}
\includegraphics[height=1.5in]{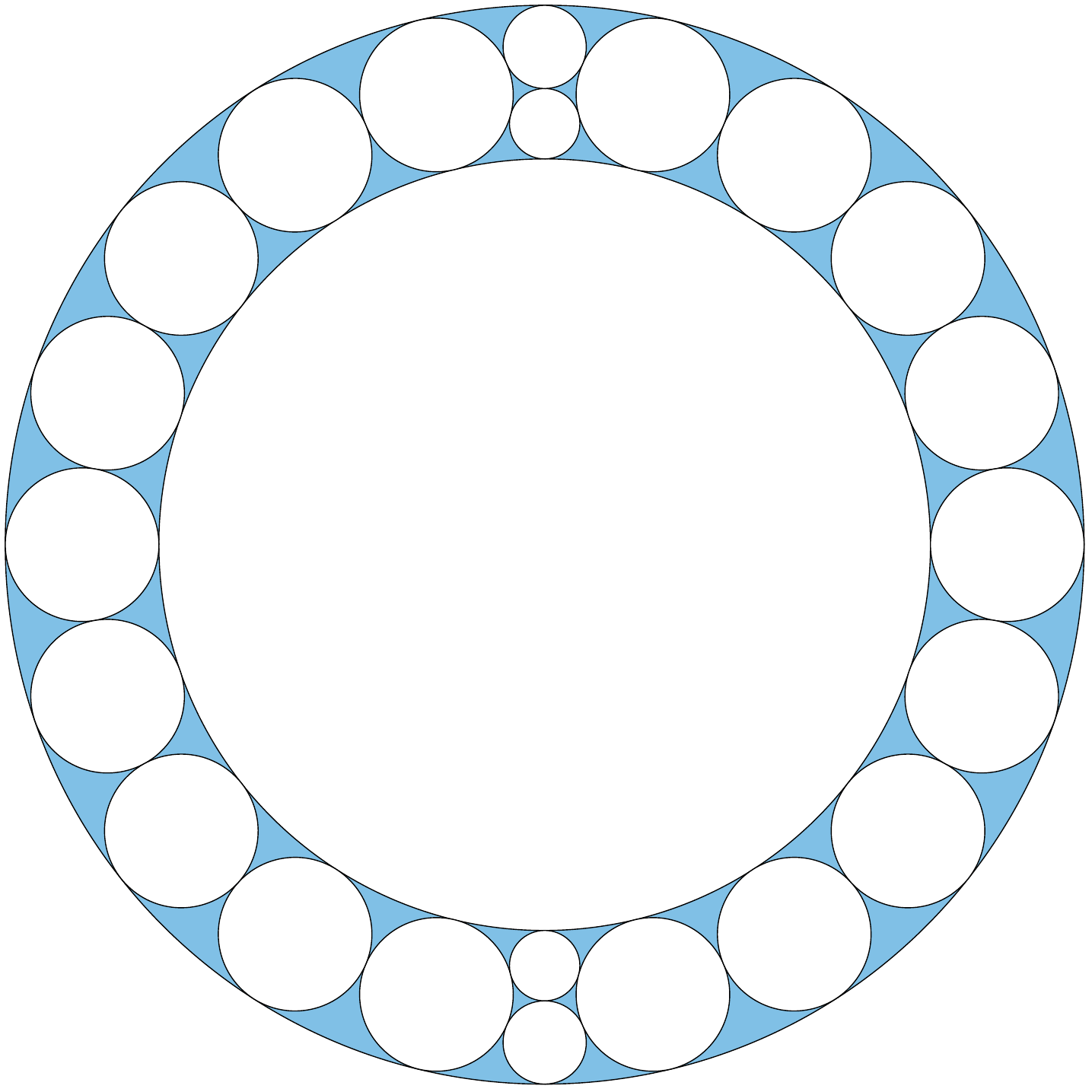}
\caption{Concentric circle packings $\Pack(2,n)$ for $n = 5,7,9,11$.}
\label{fig:pack2n}
\end{figure}

The circle packings we consider in this section are variants of the bipyramid. The graph $\Pack(k,n)$ is constructed from a cycle of length $kn + k$ by replacing every \nth{n+1} vertex with a pair of adjacent vertices. Then we add two additional vertices $u$ and $v$. As in the bipyramid, each of the vertices in the initial cycle is adjacent to both $u$ and $v$. For each pair that replaced one of the vertices in the cycle we connect one vertex in the pair to $u$ and one to~$v$. This creates a maximal planar graph with $kn+2k + 2$ vertices. The graph in Figure~\ref{fig:two-five-graph} is $\Pack(2,5)$, and Figure~\ref{fig:pack2n} depicts several additional graphs of the form $\Pack(2,n)$.

The next conjecture concerns the graphs of the form $\Pack(2,n)$. We tested its correctness for all the graphs of this form up to $n = 120$.
\begin{conjecture}
For the concentric packing of $\Pack(2,n)$, the value of $b$ satisfies a polynomial such that when
\begin{itemize}
\item $n \equiv 3,5 \pmod{6}$ its irreducible factors have Galois group: $S_{2n-2}$;
\item $n \equiv 1 \pmod{6}$ its irreducible factors have Galois groups: $S_2, S_{2n-4}$;
\item $n \equiv 0,2 \pmod{6}$ its irreducible factors have Galois groups: $S_1,S_{n-2}, S_{n-1}$;
\item $n \equiv 4 \pmod{6}$ its irreducible factors have Galois groups: $S_1$, $S_2$, $S_{n-3}$, $S_{n-2}$, $S_{n-2}$.
\end{itemize}
\end{conjecture}

\begin{figure}[t]
\centering
\includegraphics[height=1.5in]{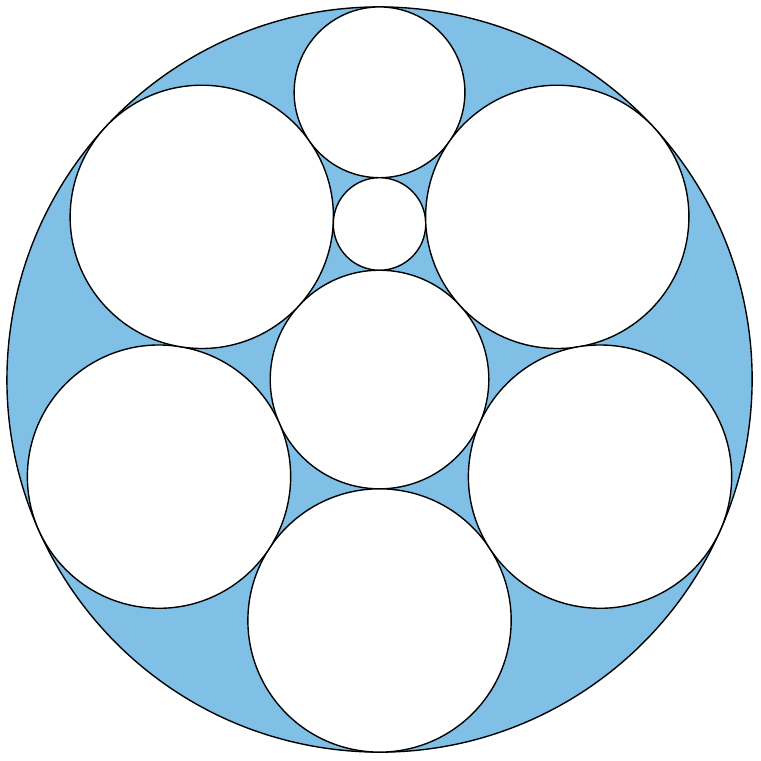}\hspace{1em}
\includegraphics[height=1.5in]{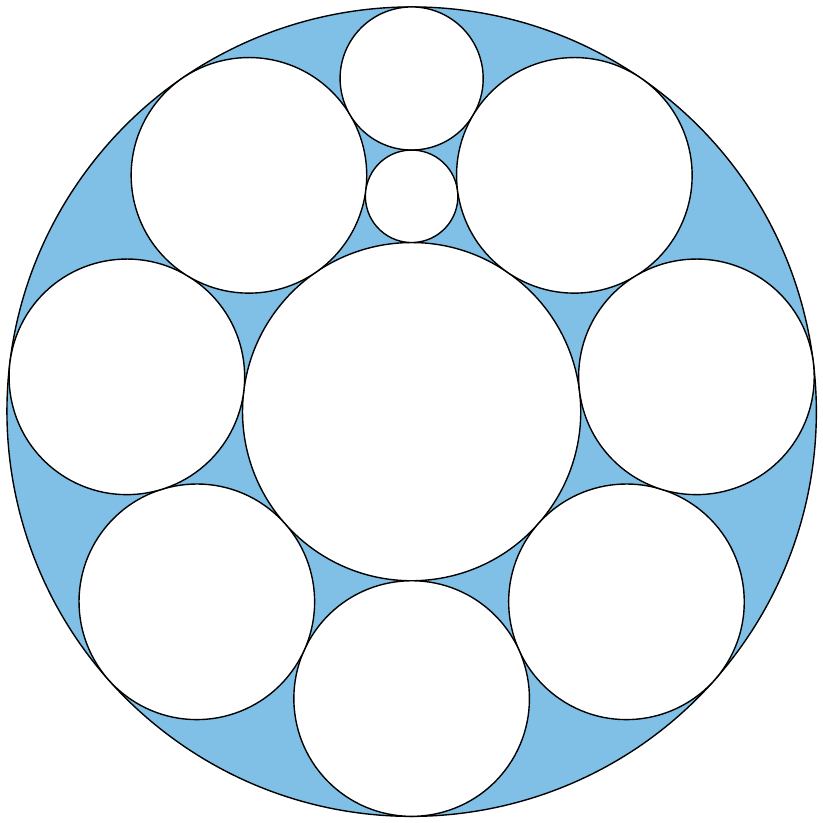}\\
\vspace{1em}
\includegraphics[height=1.5in]{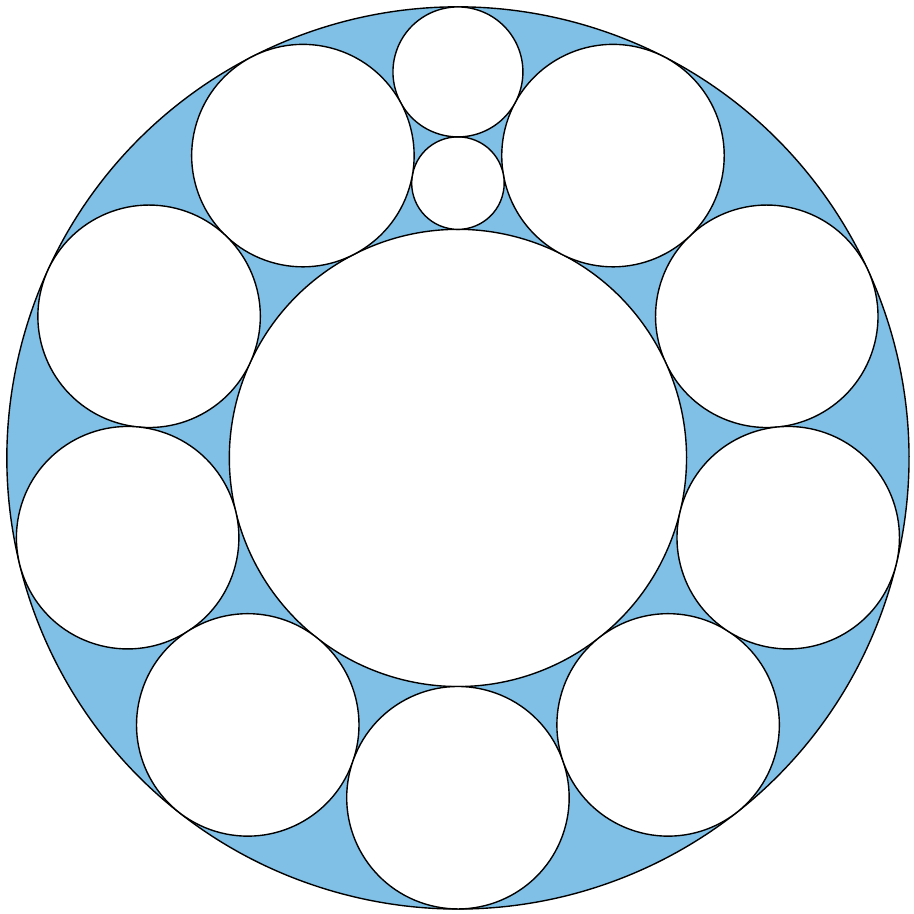}\hspace{1em}
\includegraphics[height=1.5in]{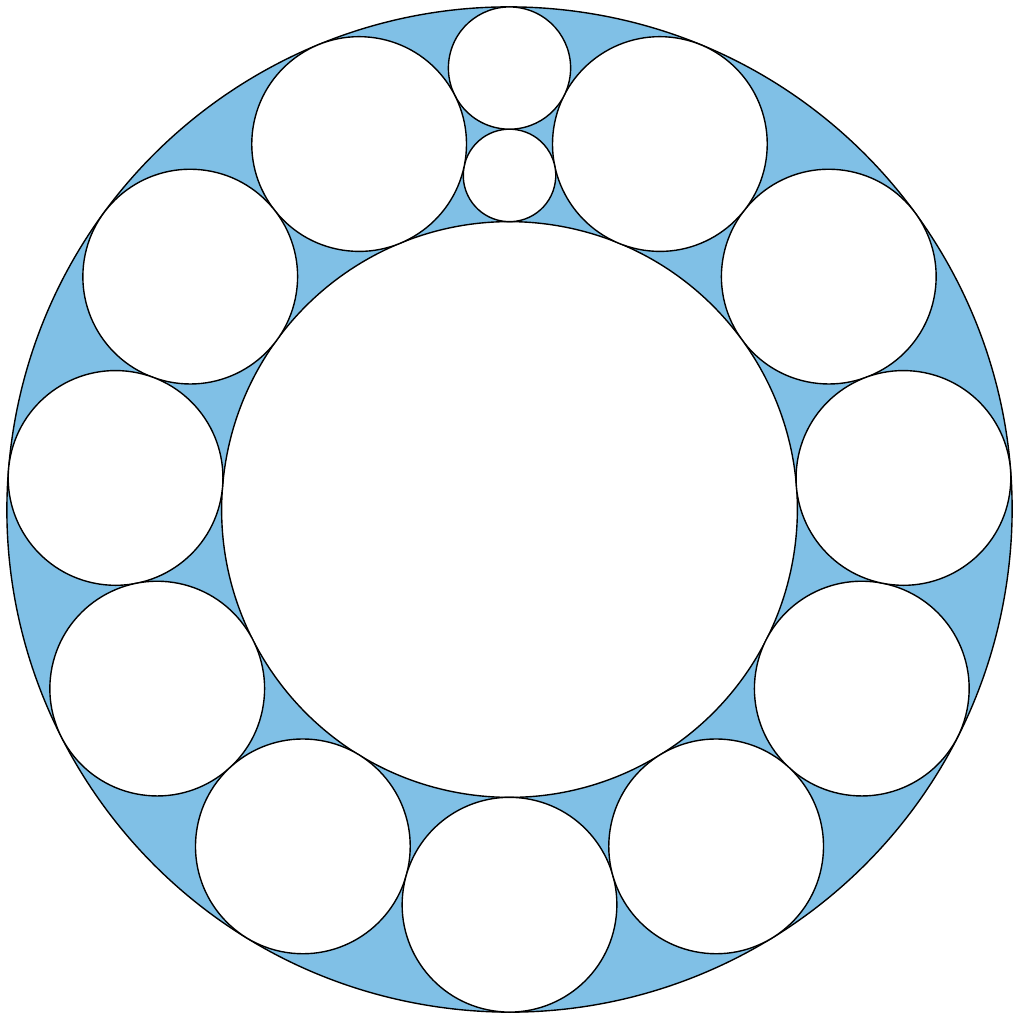}
\caption{Concentric circle packings $\Pack(1,n)$ for $n = 5,7,9,11$.}
\label{fig:pack1n}
\end{figure}

The next conjecture  concerns the graphs of the form $\Pack(1,n)$ depicted in Figure~\ref{fig:pack1n}. We tested it only up to $n = 13$ due to the increased difficulty of identifying Galois groups that are not symmetric groups. The groups $2\wr S_d$ appearing in the conjecture are the hyperoctahedral groups, symmetry groups of $d$-dimensional hypercubes.

\begin{conjecture}
For the concentric packing of $\Pack(1,n)$, the value of $b$ satisfies a polynomial such that when
\begin{itemize}
\item $n \equiv 3 \pmod{4}$ its irreducible factors have Galois groups: $2 \wr S_{n-1}$.
\item $n \equiv 1 \pmod{4}$ its irreducible factors have Galois groups: $S_1, 2 \wr S_{n-2}$;
\end{itemize}
\end{conjecture}

Either of these conjectures, if true, would imply that circle packing is hard on a \emph{root radical computation tree} that can compute both bounded-degree polynomial roots and unbounded-degree radicals. Moreover, they would imply that the degree necessary to compute circle packing on a root computation tree is linear in~$n$, without depending on the infinitude of Sophie Germain primes.

\newpage
\bibliographystyleappendix{splncs}
\bibliographyappendix{paper}

\end{appendix}

\end{document}